\documentclass[11pt]{article}

\PassOptionsToPackage{hyperfootnotes=false}{hyperref}
\usepackage[final]{acl}
\usepackage{subcaption}  
\usepackage{times}
\usepackage{latexsym}

\usepackage[T1]{fontenc}

\usepackage[utf8]{inputenc}

\usepackage{microtype}

\usepackage{inconsolata}

\usepackage{graphicx}

\usepackage{amsmath} 
\usepackage{mathtools}
\usepackage{amsfonts}
\usepackage{amsthm}
\newtheorem{proposition}{Proposition}

\usepackage{enumitem}
\usepackage{array}
\usepackage{tabularx}
\usepackage{booktabs}
\usepackage{multirow}
\usepackage{makecell}
\usepackage{threeparttable}
\usepackage{rotating}
\usepackage{bm} 
\usepackage[table]{xcolor}
\usepackage{colortbl}
\usepackage[breakable, skins, many]{tcolorbox}
\usepackage{xspace}

\usepackage{xurl} 
\usepackage{hyperref} 

\definecolor{bestcell}{RGB}{231,217,201}
\newcommand{\best}[1]{\textbf{#1}}

\renewcommand{\paragraph}[1]{\smallskip\noindent\textbf{#1.}}
\renewcommand{\subparagraph}[1]{\smallskip\noindent\textbf{\underline{#1.}}}

\newcolumntype{L}[1]{>{\raggedright\arraybackslash}p{#1}}  %

\newcommand{\modelname}{\textbf{GGSS}\xspace}

\newcommand{\codepath}[1]{\nolinkurl{#1}}

\newtcolorbox{promptbox}[1]{%
  enhanced,
  breakable,
  colback=gray!8,
  colframe=gray!40,
  colbacktitle=gray!70!black,
  coltitle=white,
  fonttitle=\bfseries,
  title=#1,
  boxrule=0.6pt,
  arc=2mm,
  left=3.5mm,right=3.5mm,top=2.5mm,bottom=2.5mm,
  fontupper=\small\ttfamily\raggedright\sloppy,
}

\title{GGSS: Geodesic-Gated Spherical Steering for Inference-Time Debiasing of Generative Vision--Language Models}

\author{
  \textbf{Yiqun Sun\textsuperscript{1}},
  \textbf{Junyu Chen\textsuperscript{2}},
  \textbf{Pengfei Wei\textsuperscript{1}}\thanks{\,Corresponding author.},
  \textbf{Lawrence B. Hsieh\textsuperscript{1}}
\\
  \textsuperscript{1}Magellan Technology Research Institute (MTRI),
  \textsuperscript{2}National University of Singapore
\\
  \texttt{\{duke.sun, pengfei.wei, lawrence.hsieh\}@mtri.co.jp},
  \texttt{chenjunyu@u.nus.edu}
}

\begin{document}

\renewcommand{\thefootnote}{\fnsymbol{footnote}}
\maketitle
\renewcommand{\thefootnote}{\arabic{footnote}}
\setcounter{footnote}{0}
\begin{abstract}
Generative vision--language models (VLMs) are increasingly used in human-centered settings, yet they can produce demographically biased outputs even when images differ only in controlled attributes such as perceived race or gender. However, existing inference-time debiasers were largely designed for static embeddings or CLIP-like models rather than generative VLMs. We propose \modelname---\textbf{G}eodesic-\textbf{G}ated \textbf{S}pherical \textbf{S}teering---a norm-preserving intervention that discovers a counterfactual bias subspace on the unit hypersphere, steers visual tokens along geodesic arcs, and uses an adaptive gate to focus correction on tokens that carry stronger demographic signal. We evaluate four generative VLMs against ten adapted inference-time debiasing baselines and prompt-based mitigation under a single operating-point protocol across categorical, pairwise, and occupation-gender bias tests, while also measuring general visual-language capability. \modelname\ achieves the lowest average bias on all four models, significant on three of four backbones under paired permutation tests, while preserving MMStar accuracy within $\pm 0.6$\,p.p.\ of the unsteered baseline. Code is available at \url{https://github.com/dukesun99/GGSS}.
\end{abstract}

\section{Introduction}
\label{sec:intro}

\begin{figure}[t]
  \centering
  \includegraphics[width=0.99\columnwidth]{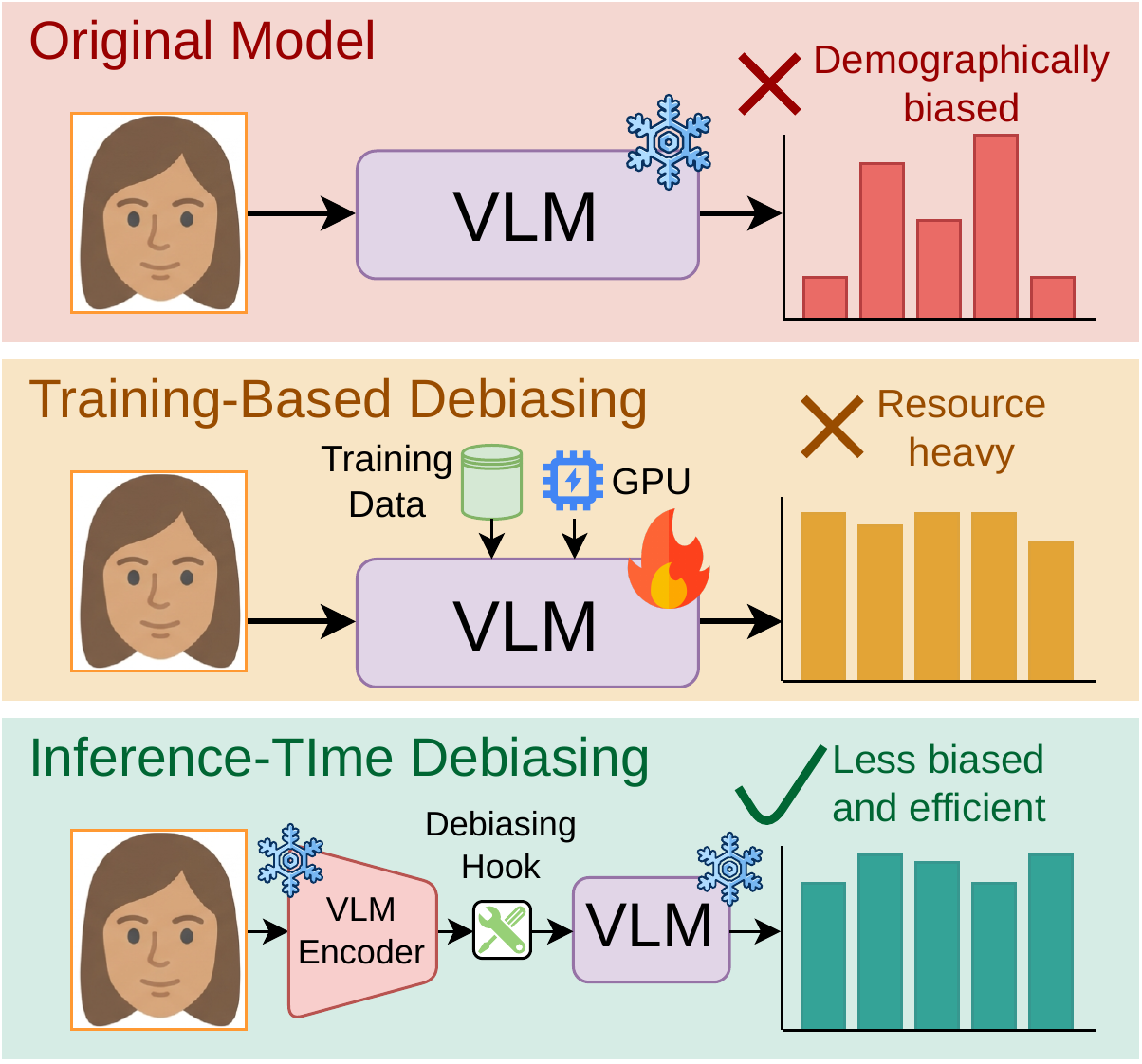}
  \caption{\textbf{Motivation for inference-time debiasing.} Instead of retraining a VLM with additional data and heavy GPU resources, GGSS inserts a lightweight debiasing hook during inference to reduce demographic bias efficiently while keeping the original model frozen.}
  \label{fig:teaser}
  \vspace{-1.0em}
\end{figure}

Vision--language models (VLMs) now appear in human-centered, and sometimes high-stakes, decision-making settings, including recruitment and decision support \cite{liu2023visual,wu2024candidate,banerjee2024learning}. As they become more capable and more widely deployed, a growing body of work has documented that modern VLMs produce systematically different outputs across demographic groups even when the underlying visual content is controlled \cite{chen2026reflect,howard2024socialcounterfactuals,huang2025visbias}. Debiasing such models at \emph{inference time} (Figure~\ref{fig:teaser}) is attractive because it avoids retraining, applies to frozen checkpoints, and can be reconfigured as the target notion of bias changes \cite{chuang2023biasedprompts,gerych2024bendvlm,wring2026}. However, many existing inference-time debiasing methods were developed for CLIP-like \cite{radford2021learning} settings where an input is summarized by a single global embedding. INLP, R-LACE, and LEACE erase concepts through linear subspace operations \cite{ravfogel2020null,ravfogel2022rlace,belrose2023leace}. Classifier-guided methods apply mean-shift corrections. BendVLM performs retrieval-based equalization in CLIP-style embedding spaces \cite{gerych2024bendvlm,radford2021learning}. Modern \emph{generative} VLMs \cite{bai2025qwen3,liu2024llavanext} instead encode an image as many visual tokens that are fused with a decoder-only language model. Reusing CLIP-era methods in this setting can be overly destructive and, as our experiments show, can even be dominated by the unsteered baseline.

We identify three concrete challenges that a successful inference-time debiaser for generative VLMs must address. First, \textbf{geometry preservation}: visual activations in recent VLMs enter the language model through a small set of projection layers whose outputs are later consumed by attention and decoding. Euclidean subtraction of a bias direction distorts both directions and norms; empirically this matters increasingly with model scale, causing sharp drops in general quality and over-steering failure modes that on-sphere steering avoids (\S\ref{sec:ablation}). Second, \textbf{token-level bias heterogeneity}: demographic signal in a generative VLM is concentrated in a small fraction of visual tokens (e.g., face or clothing tokens) while many tokens encode pose, background, or textual cues and carry essentially none. Hard subspace projection applied uniformly over-corrects the latter and silently degrades task performance. Third, \textbf{single-direction rigidity}: many existing methods assume a binary protected attribute and a single ``bias direction,'' but race-like attributes in practice have multi-class structure and must be handled as a low-dimensional subspace \cite{manzini2019black}.

We propose \modelname\ (\textbf{G}eodesic-\textbf{G}ated \textbf{S}pherical \textbf{S}teering), an inference-time debiasing framework that matches these three challenges with three coordinated design choices. First, we discover a \emph{bias subspace} $\mathbf{V}_{\mathrm{bias}}$ from counterfactual image sets that vary demographic attributes while holding other visual factors as fixed as possible. We collect pooled visual activations and apply singular value decomposition to tangent-space shifts around per-group Fr\'echet means on the unit hypersphere, yielding a multi-dimensional and norm-aware estimate of demographic variation. Second, at inference time we replace hard null-space projection with \emph{spherical linear interpolation} \citep{shoemake1985slerp} between the original token direction and a debiased target, producing a smooth geodesic rotation that preserves norms exactly. Third, we add an \emph{adaptive confidence gate} $g_i$ per token, calibrated from the distribution of bias-projection magnitudes observed during discovery. This gate applies little or no steering to low-bias tokens and stronger steering to tokens with unusually large bias components, which makes the intervention more selective than a uniform projection.

We evaluate \modelname\ on four generative VLMs against ten adapted inference-time debiasing baselines, a prompt-based mitigation family, and four structural ablations. The evaluation covers race- and gender-related bias tests together with MMStar capability preservation, and all methods are compared under a fair protocol, with bootstrap confidence intervals and paired permutation tests for every reported comparison. Experiments show \modelname\ attains the lowest average bias on all four models, with statistically significant reductions on three of the four backbones and per-task reductions up to 96\%, 84\%, and 61\%, while keeping MMStar accuracy within $\pm 0.6$\,p.p.\ of baseline. These results support our central claim: for generative VLMs, \emph{adaptive} geodesic steering is a better primitive for inference-time bias mitigation than hard, uniform projection. Matched-gate ablations attribute robustness at scale to on-sphere steering and selectivity to the gate.

\medskip

\noindent Our contributions are fourfold:
\begin{itemize}[nolistsep,left=8pt]
    \item We identify why hard, token-uniform, single-direction debiasing operators under-perform on generative VLMs, and propose \modelname\ as a norm-preserving geodesic alternative.
    \item We introduce an adaptive gate that uses the discovery-set bias-norm distribution to assign per-token steering strengths without additional training.
    \item We port representative CLIP- and LM-era debiasers to the same generative-VLM intervention layer, producing a unified comparison suite.
    \item We evaluate four VLMs under one best-avg-$\alpha$ protocol; \modelname\ achieves the lowest average bias on all four while preserving MMStar capability.
\end{itemize}

\section{Related Work}
\label{sec:related_work}

\paragraph{Social biases in vision--language models}
Prior work shows that VLMs encode and express social biases across representations, benchmarks, and downstream behaviors. Grounded vision-language embeddings inherit demographic associations from earlier representation models~\cite{ross2021groundedbias,srinivasan2022dangerous}, and later studies measure disparities across gender, race, age, nationality, religion, and other attributes~\cite{zhou2022vlstereoset,janghorbani2023mmbias,ruggeri2023multidim}. These biases also surface in captioning, visual question answering, retrieval, open-ended prompting, counterfactual evaluation, and real-image settings~\cite{hirota2022captionbias,hall2023zeroshot,ghate2025propagate,howard2024socialcounterfactuals,wang2024vlbiasbench,huang2025visbias,narnaware2026bbqvbenchmarkingvisualstereotype}. This motivates inference-time mitigation methods that can be applied to frozen models without retraining.

\paragraph{VLM debiasing: training-based vs.\ inference-time}
Most VLM debiasing methods target CLIP-style dual encoders, not generative VLMs. Training-based approaches learn lightweight transformations, pruning or imputation modules, causal adjustments, bias-corpus disentanglement, LLM-guided projections, or joint image--text corrections~\cite{seth2023dear,jung2024sfid,pang2025cabin,jang2025targetbias,molahasani2025prism,zhang2025jointvl}. Inference-time methods instead modify prompts or embedding geometry at test time, including calibrated prompt-defined projection, query-specific local debiasing, and rotation-based correction~\cite{chuang2023biasedprompts,gerych2024bendvlm,wring2026}. Post-hoc transformations of embedding geometry also serve instruction following~\cite{feng2025don} and concept suppression in text-to-image generation~\cite{chen2025growth,li2025responsible}. Debiasing for generative VLMs is less explored, with recent work moving toward internal activation intervention and monitoring~\cite{ratzlaff2024ablating,an2026interpretable,huben2024sae,cheng2025social,singh2026prompts,wang2026sparse}. \modelname\ follows this activation-level direction, but focuses on generative VLM on token activations.

\paragraph{Concept erasure and activation steering}
Our baselines draw on linear concept erasure and inference-time activation steering. INLP, R-LACE, and LEACE remove protected information by learning or solving for linear subspace projections~\cite{ravfogel2020null,ravfogel2022rlace,belrose2023leace}. These classical erasers are useful reference points, but they were designed for static embeddings or encoder representations and typically act as hard, token-uniform projections. Activation steering methods such as ActAdd, Representation Engineering, and Inference-Time Intervention instead manipulate hidden states directly at test time~\cite{turner2023activation,zou2023representation,li2023inferencetime}. \modelname\ shares the inference-time spirit of this line, but replaces additive Euclidean steering with spherical geometry. The unit hypersphere and spherical interpolation are standard tools in directional statistics and computer graphics~\cite{mardia2000directional,pennec2006riemannian,shoemake1985slerp}, and related geometric views of concept directions have appeared in representation analysis and null-space steering~\cite{park2023linear,sheng2026alphasteer,sun2025one}. We instantiate these ideas for image-side activations in generative VLMs and add a bias-norm-calibrated gate that makes steering token-adaptive.

\section{The GGSS Framework}
\label{sec:framework}

\begin{figure*}[ht]
  \centering
  \includegraphics[width=\textwidth]{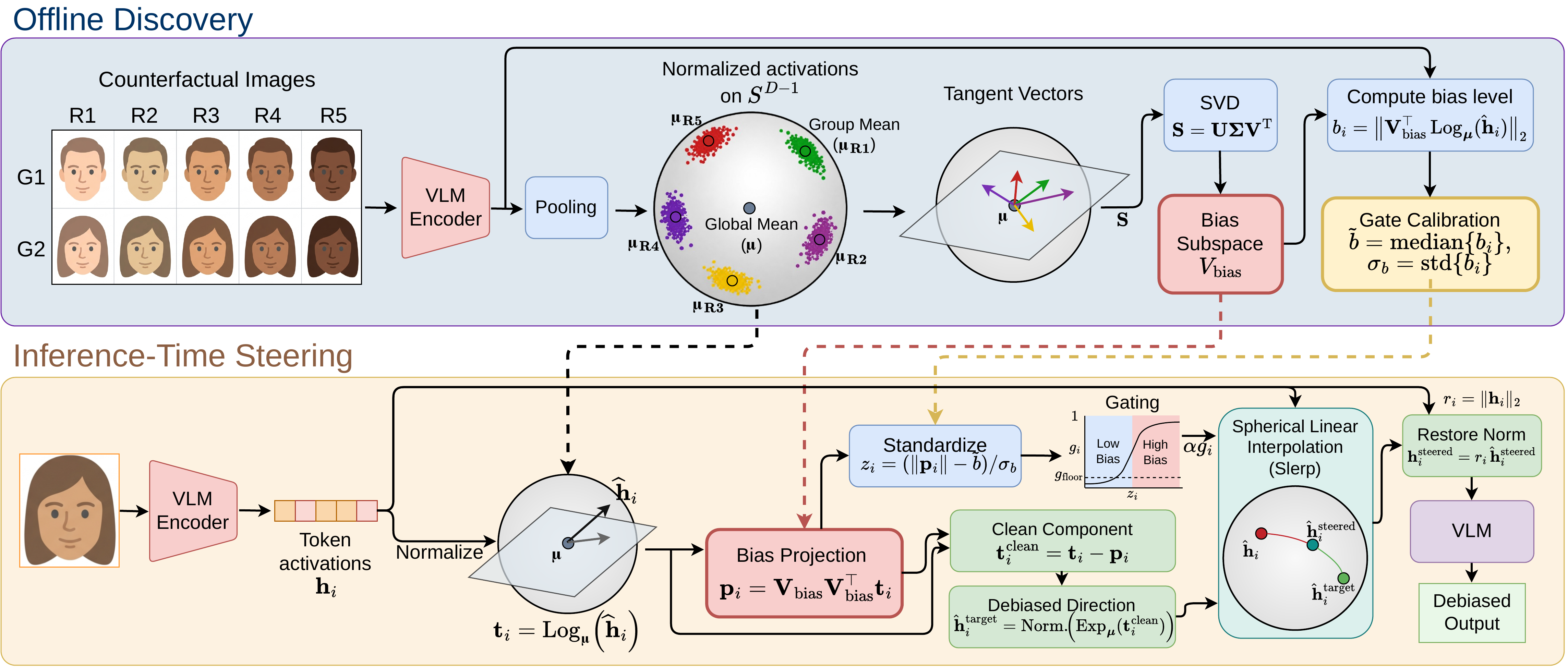}
  \caption{\textbf{Overview of the \modelname\ framework.} The offline stage estimates a spherical bias subspace from counterfactual image activations and calibrates a bias-sensitive token gate. At inference time, each visual token is mapped to tangent space, decomposed into bias and clean components, gated according to its bias magnitude, geodesically rotated toward a debiased direction via Slerp, and finally norm-restored before being passed to the LLM decoder of the VLM.}
  \label{fig:overview}
  \vspace{-1.0em}
\end{figure*}

We propose \modelname---\textbf{G}eodesic-\textbf{G}ated \textbf{S}pherical \textbf{S}teering---a two-stage framework for debiasing generative VLMs, summarized in Figure~\ref{fig:overview}. Offline discovery passes counterfactual images through the frozen VLM to estimate the steering quantities: a protected-attribute basis, a global spherical reference point $\boldsymbol{\mu}$, and gate calibration statistics $(\tilde{b},\sigma_b)$. Online inference is separate: a new test input produces an activation tensor, and \modelname\ steers its visual tokens using the discovered quantities. The learned quantities are reused for new inputs at inference time.

\paragraph{Notation}
Let
\[
\mathbb{S}^{D-1}
=
\{\mathbf{x}\in\mathbb{R}^D:\|\mathbf{x}\|_2=1\}
\]
be the unit sphere in $\mathbb{R}^D$, and let
$\nu(\mathbf{x})
=
\frac{\mathbf{x}}{\|\mathbf{x}\|_2}$ be the normalization operator. For $\mathbf{p},\mathbf{q}\in\mathbb{S}^{D-1}$, define $d(\mathbf{p},\mathbf{q})=\arccos\langle\mathbf{p},\mathbf{q}\rangle$ and
\[
\begin{aligned}
\log_{\mathbf{p}}(\mathbf{q})
&=
d(\mathbf{p},\mathbf{q})
\frac{\mathbf{q}-\langle\mathbf{p},\mathbf{q}\rangle\mathbf{p}}
{\|\mathbf{q}-\langle\mathbf{p},\mathbf{q}\rangle\mathbf{p}\|_2},
\\
\exp_{\mathbf{p}}(\mathbf{t})
&=
\cos(\|\mathbf{t}\|_2)\mathbf{p}
+
\sin(\|\mathbf{t}\|_2)\frac{\mathbf{t}}{\|\mathbf{t}\|_2},
\end{aligned}
\]
with the standard conventions $\log_{\mathbf{p}}(\mathbf{p})=\mathbf{0}$ and $\exp_{\mathbf{p}}(\mathbf{0})=\mathbf{p}$. For $\mathbf{p},\mathbf{q}\in\mathbb{S}^{D-1}$ with $\mathbf{p}\neq\mathbf{q}$ and $\mathbf{q}\neq-\mathbf{p}$, define spherical linear interpolation by
\[
\begin{aligned}
\operatorname{Slerp}(\mathbf{p},\mathbf{q};\beta)
&=
\frac{\sin((1-\beta)\theta)}{\sin\theta}\mathbf{p}
+
\frac{\sin(\beta\theta)}{\sin\theta}\mathbf{q},
\\
\theta&=\arccos\langle\mathbf{p},\mathbf{q}\rangle.
\end{aligned}
\]
When $\mathbf{p}=\mathbf{q}$, we use the continuous extension $\operatorname{Slerp}(\mathbf{p},\mathbf{p};\beta)=\mathbf{p}$. For $0\leq\beta\leq 1$, Slerp interpolates between $\mathbf{p}$ and $\mathbf{q}$; for $\beta>1$, it continues along the same geodesic beyond the target.

The rest of this section follows the
two stages of GGSS. Section~\ref{sec:discovery} presents the offline
discovery stage, where counterfactual activations are used to estimate a
protected-attribute subspace, a global reference point, and gate calibration
statistics. Section~\ref{sec:steering} presents the inference-time
steering stage, where each visual token is projected away from the learned
protected-attribute subspace, updated along the sphere, and rescaled to its
original norm.

\subsection{Counterfactual Bias Subspace Discovery}
\label{sec:discovery}

\paragraph{Counterfactual representatives}
We write the discovery images as
\[
\mathbf{x}_{u,a},\qquad u\in\mathcal{U},\quad a\in\mathcal{A}.
\]
Here $a$ indexes the protected attribute value, such as a race category, while $u$ collects the fixed factors, such as identity, occupation, gender, pose, clothing, background, and lighting. This abstracts the concrete occupation/identity/race/gender indexing specified in Appendix~\ref{app:notation}. Running the frozen model on $\mathbf{x}_{u,a}$ with a fixed \emph{discovery prompt}, the text instruction paired with each discovery image (``\emph{Describe this image in detail.}''), gives the target-layer activation
\[
\mathbf{H}_{u,a}\in\mathbb{R}^{T\times D}.
\]
The choice of discovery prompt is immaterial at our intervention layer: the vision-to-language projection computes its visual-token activations before any interaction with the text prompt, so the prompt architecturally cannot influence them, and re-running discovery with four alternative prompts yields subspaces identical to numerical precision (principal-angle cosines $\geq 0.9997$) and identical end-to-end results (Appendix~\ref{app:prompt_invariance}).
We summarize each image by the normalized pooled representative
\[
\mathbf{y}_{u,a}
=
\nu\left(\frac{1}{T}\sum_{t=1}^{T}\mathbf{H}_{u,a,t}\right)
\in\mathbb{S}^{D-1}.
\]
Within a fixed context $u$, varying $a$ is designed to emphasize protected-attribute variation while keeping other visual factors fixed. Pooling yields one stable spherical representative per counterfactual image.

\paragraph{Subspace estimation}
For each fixed context $u$, let the spherical Fr\'echet mean be
\[
\mathbf{c}_u
\in
\arg\min_{\mathbf{z}\in\mathbb{S}^{D-1}}
\sum_{a\in\mathcal{A}} d(\mathbf{z},\mathbf{y}_{u,a})^2.
\]
The protected-attribute shift for value $a$ is
\[
\mathbf{s}_{u,a}=\log_{\mathbf{c}_u}(\mathbf{y}_{u,a}).
\]
Stacking the shifts $\{\mathbf{s}_{u,a}:u\in\mathcal{U},a\in\mathcal{A}\}$ as rows gives
the shift matrix
\[
\mathbf{S}\in\mathbb{R}^{N_{\mathrm{shifts}}\times D},
\qquad
N_{\mathrm{shifts}}=|\mathcal{U}||\mathcal{A}|.
\]
Each row of $\mathbf{S}$ is one tangent shift induced by changing the protected
attribute within a fixed context. We then apply SVD:
\begin{equation}
\mathbf{S}=\mathbf{U}\boldsymbol{\Sigma}\mathbf{V}^{\top},
\qquad
\mathbf{V}_{\mathrm{bias}}=[\boldsymbol{v}_1,\ldots,\boldsymbol{v}_k]\in\mathbb{R}^{D\times k}.
\label{eq:ggss-svd-basis}
\end{equation}
Here $\boldsymbol{v}_j$ is the $j$-th column of $\mathbf{V}$. Thus $\mathbf{V}_{\mathrm{bias}}$ is the matrix of the top-$k$ right singular vectors of $\mathbf{S}$, and $\operatorname{span}(\mathbf{V}_{\mathrm{bias}})$ is the learned protected-attribute subspace. Since the columns of $\mathbf{V}$ are orthonormal, $\mathbf{V}_{\mathrm{bias}}^{\top}\mathbf{V}_{\mathrm{bias}}=I_k$. We set $k=|\mathcal{A}|-1$ by default. The global reference point is the Fr\'echet mean of all discovery representatives:
\[
\boldsymbol{\mu}
\in
\arg\min_{\mathbf{z}\in\mathbb{S}^{D-1}}
\sum_{u\in\mathcal{U}}\sum_{a\in\mathcal{A}}
d(\mathbf{z},\mathbf{y}_{u,a})^2.
\]
Together, $\mathbf{V}_{\mathrm{bias}}$ and $\boldsymbol{\mu}$ parameterize the steering geometry. Unlike classifier-based discovery methods, the subspace is estimated directly from counterfactual activation differences and requires no supervised probe.

\paragraph{Gate calibration statistics}
\label{sec:token-steering}
For each discovery representative $\mathbf{c}\in\{\mathbf{y}_{u,a}:u\in\mathcal{U},a\in\mathcal{A}\}$, compute
\[
b(\mathbf{c})=\|\mathbf{V}_{\mathrm{bias}}^{\top}\log_{\boldsymbol{\mu}}(\mathbf{c})\|_2,
\]
where $\mathbf{V}_{\mathrm{bias}}^{\top}\log_{\boldsymbol{\mu}}(\mathbf{c})$ gives the coordinates of $\mathbf{c}$, after mapping to the reference point $\boldsymbol{\mu}$, along the learned protected-attribute basis. Thus $b(\mathbf{c})$ is the protected-coordinate magnitude of the discovery representative. We store
\[
\tilde{b}=\operatorname{median}\{b(\mathbf{c})\},
\qquad
\sigma_b=\operatorname{std}\{b(\mathbf{c})\}.
\]
These scalars calibrate the token gate at inference time.

\subsection{Token-Level Geodesic Steering}
\label{sec:steering}

At inference time, the new input produces a target-layer activation tensor
\[
\mathbf{Z}\in\mathbb{R}^{n_{\mathrm{batch}}\times T'\times D}.
\]
We flatten it to token vectors $\{\mathbf{h}_i\}_{i=1}^{N}$, with $N=n_{\mathrm{batch}}T'$, and process each nonzero token independently. Let
\[
r_i=\|\mathbf{h}_i\|_2,
\qquad
\widehat{\mathbf{h}}_i=\nu(\mathbf{h}_i),
\qquad
\mathbf{t}_i=\log_{\boldsymbol{\mu}}(\widehat{\mathbf{h}}_i).
\]
The radius $r_i$ is kept outside the debiasing step and restored at the end.

\paragraph{Projection and target direction}
Because $\mathbf{V}_{\mathrm{bias}}$ has orthonormal columns, $\mathbf{V}_{\mathrm{bias}}\mathbf{V}_{\mathrm{bias}}^{\top}$ is the orthogonal projector onto $\operatorname{span}(\mathbf{V}_{\mathrm{bias}})$. The protected-coordinate component and its orthogonal complement are
\begin{equation}
\begin{aligned}
\mathbf{p}_i
&=
\mathbf{V}_{\mathrm{bias}}\mathbf{V}_{\mathrm{bias}}^{\top}\mathbf{t}_i,
\\
\mathbf{t}_i^{\mathrm{clean}}
&=
(I-\mathbf{V}_{\mathrm{bias}}\mathbf{V}_{\mathrm{bias}}^{\top})\mathbf{t}_i.
\end{aligned}
\label{eq:ggss-projection}
\end{equation}
Here $\mathbf{p}_i$ is the component removed by \modelname, while $\mathbf{t}_i^{\mathrm{clean}}$ is the steering coordinate after zeroing the learned protected-attribute coordinates.
Mapping the projected tangent vector back to the sphere gives the target direction
\[
\widehat{\mathbf{h}}_i^{\mathrm{target}}
=
\exp_{\boldsymbol{\mu}}(\mathbf{t}_i^{\mathrm{clean}}).
\]
This defines the target direction by first projecting in the steering coordinates and then mapping the projected coordinate back to the sphere.

\paragraph{Calibrated token gate}
The token-specific gate is
\begin{equation}
\begin{aligned}
z_i&=
\frac{\|\mathbf{V}_{\mathrm{bias}}^{\top}\mathbf{t}_i\|_2-\tilde{b}}
{\sigma_b+\varepsilon},
\\
g_i&=
g_{\mathrm{floor}}
+
(1-g_{\mathrm{floor}})
\operatorname{sigmoid}(\kappa z_i),
\end{aligned}
\label{eq:ggss-gate}
\end{equation}
where $\varepsilon>0$ avoids division by zero, $\kappa$ controls gate sharpness, and $g_{\mathrm{floor}}$ sets the minimum steering strength. Since $\mathbf{V}_{\mathrm{bias}}$ has orthonormal columns, $\|\mathbf{V}_{\mathrm{bias}}^{\top}\mathbf{t}_i\|_2=\|\mathbf{p}_i\|_2$. Thus the gate compares the token's protected-coordinate magnitude with the discovery statistics $(\tilde{b},\sigma_b)$. Tokens with larger protected-coordinate magnitude receive stronger steering, while low-bias tokens stay closer to their original directions.

\paragraph{Geodesic update}
Set $\beta_i=\alpha g_i$ and update the token direction by
\begin{equation}
\widehat{\mathbf{h}}_i^{\mathrm{steered}}
=
\operatorname{Slerp}(\widehat{\mathbf{h}}_i,\widehat{\mathbf{h}}_i^{\mathrm{target}};\beta_i),
\label{eq:ggss-slerp-update}
\end{equation}
then restore the original radius:
\begin{equation}
\mathbf{h}_i^{\mathrm{steered}}
=
r_i\widehat{\mathbf{h}}_i^{\mathrm{steered}}.
\label{eq:ggss-norm-restore}
\end{equation}
When $\beta_i=0$ the direction is unchanged, while $\beta_i=1$ reaches the projected target. Larger values continue along the same geodesic, allowing the steering strength $\alpha$ to control how aggressively the protected-coordinate component is reduced.

The following proposition records the mechanical guarantees of the update:
GGSS removes the learned protected-attribute coordinates by an orthogonal
projection in the steering coordinates, and the final spherical update
preserves the original token norm exactly. Whether the steering reduces bias
is an empirical question, addressed in \S\ref{sec:experiments}.

\begin{proposition}
\label{prop:ggss-projection-norm}
For the basis $\mathbf{V}_{\mathrm{bias}}$ defined in Eq.~\eqref{eq:ggss-svd-basis}
and the token update defined in
Eqs.~\eqref{eq:ggss-projection}--\eqref{eq:ggss-norm-restore}, the projected
coordinate $\mathbf{t}_i^{\mathrm{clean}}$ is the Euclidean projection of $\mathbf{t}_i$ onto
the set of vectors whose protected-attribute coordinates are zero:
\[
\mathbf{t}_i^{\mathrm{clean}}
=
\arg\min_{\mathbf{z}\in\mathbb{R}^{D}}
\left\{
\|\mathbf{z}-\mathbf{t}_i\|_2^2:
\mathbf{V}_{\mathrm{bias}}^{\top}\mathbf{z}=\mathbf{0}
\right\}.
\]
Moreover, if $\widehat{\mathbf{h}}_i^{\mathrm{target}}\neq -\widehat{\mathbf{h}}_i$, with the convention $\operatorname{Slerp}(\mathbf{p},\mathbf{p};\beta)=\mathbf{p}$, then
\[
\|\mathbf{h}_i^{\mathrm{steered}}\|_2=\|\mathbf{h}_i\|_2.
\]
\end{proposition}

The proof is given in Appendix~\ref{app:proof-ggss-projection-norm}.

The proposition shows that GGSS removes the learned protected-attribute coordinates by an orthogonal projection and preserves the token norm after the spherical update.

The steered tokens are reshaped back into the activation tensor and passed to the remaining layers of the frozen VLM. Because the gate is token-dependent, low-bias tokens are changed less, while tokens with larger protected-coordinate magnitude are steered more strongly.

\section{Experiments}
\label{sec:experiments}

\begin{table*}[!ht]
\centering
\setlength{\tabcolsep}{2.5pt}
\renewcommand{\arraystretch}{1.05}
\scriptsize

\providecommand{\dual}[2]{\begin{tabular}{@{}r@{}}#1\\[-1.5pt]{\tiny\color{gray!55!black}#2}\end{tabular}}
\providecommand{\dualb}[2]{\begin{tabular}{@{}r@{}}{\bfseries\boldmath #1}\\[-1.5pt]{\tiny\color{gray!55!black}#2}\end{tabular}}
\providecommand{\dualo}[2]{\begin{tabular}{@{}r@{}}{\bfseries\boldmath #1}\\[-1.5pt]{\tiny\color{gray!55!black}{\bfseries\boldmath #2}}\end{tabular}}
\providecommand{\dualob}[2]{\begin{tabular}{@{}r@{}}{\bfseries\boldmath #1}\\[-1.5pt]{\tiny\color{gray!55!black}{\bfseries\boldmath #2}}\end{tabular}}
\providecommand{\avg}[1]{#1}
\providecommand{\avgb}[1]{{\bfseries\boldmath #1}}
\providecommand{\avgo}[1]{{\bfseries\boldmath #1}}
\providecommand{\avgob}[1]{{\bfseries\boldmath #1}}
\providecommand{\invalidrun}{\ensuremath{^\ast}}

\resizebox{\textwidth}{!}{%
\begin{tabular}{@{}l rrr r rrr r rrr r rrr r@{}}
\toprule[1.2pt]
& \multicolumn{4}{c}{\textbf{Pixtral-12B}} & \multicolumn{4}{c}{\textbf{LLaVA-1.6-Vicuna-7B}} & \multicolumn{4}{c}{\textbf{LLaVA-1.6-Mistral-7B}} & \multicolumn{4}{c}{\textbf{Qwen3-VL-4B}} \\
\cmidrule(lr){2-5}\cmidrule(lr){6-9}\cmidrule(lr){10-13}\cmidrule(lr){14-17}
\textbf{Method}
  & \makecell{MCQ\\$\times 10^3 \downarrow$}
  & \makecell{2AFC\\$\downarrow$}
  & \makecell{N/D\\$\downarrow$}
  & \makecell{Avg\\$\Delta\%\downarrow$}
  & \makecell{MCQ\\$\times 10^3 \downarrow$}
  & \makecell{2AFC\\$\downarrow$}
  & \makecell{N/D\\$\downarrow$}
  & \makecell{Avg\\$\Delta\%\downarrow$}
  & \makecell{MCQ\\$\times 10^3 \downarrow$}
  & \makecell{2AFC\\$\downarrow$}
  & \makecell{N/D\\$\downarrow$}
  & \makecell{Avg\\$\Delta\%\downarrow$}
  & \makecell{MCQ\\$\times 10^3 \downarrow$}
  & \makecell{2AFC\\$\downarrow$}
  & \makecell{N/D\\$\downarrow$}
  & \makecell{Avg\\$\Delta\%\downarrow$} \\
\midrule
Baseline  & 25.75 & 0.243 & 0.425 & $0\%$ & 8.70 & 0.000 & 0.600 & $0\%$ & 26.87 & 0.000 & 0.625 & $0\%$ & 14.65 & 0.368 & 0.400 & $0\%$ \\
\midrule
\multicolumn{17}{@{}l@{}}{\emph{INLP \cite{ravfogel2020null}}} \\
INLP (Eucl.)
  & \dual{15.05}{$-42\%$}
  & \dual{0.131}{$-46\%$}
  & \dual{0.350}{$-18\%$}
  & \avg{$-35\%$}
  & \dual{1.87}{$-79\%$}
  & \dual{0.000}{--}
  & \dual{0.625}{$+4\%$}
  & \avg{$-37\%$}
  & \dual{10.52}{$-61\%$}
  & \dual{0.000}{--}
  & \dual{0.375}{$-40\%$}
  & \avg{$-50\%$}
  & \dual{5.71}{$-61\%$}
  & \dual{0.279}{$-24\%$}
  & \dual{0.375}{$-6\%$}
  & \avg{$-31\%$} \\
INLP (sph.)
  & \dualb{0.47}{$-98\%$}
  & \dual{0.297}{$+22\%$}
  & \dual{0.350}{$-18\%$}
  & \avg{$-31\%$}
  & \dual{2.05}{$-76\%$}
  & \dual{0.000}{--}
  & \dualb{0.025}{$-96\%$}
  & \avg{$-86\%$}
  & \dual{9.87}{$-63\%$}
  & \dual{0.000}{--}
  & \dualb{0.050}{$-92\%$}
  & \avg{$-78\%$}
  & \dual{11.00}{$-25\%$}
  & \dual{0.238}{$-35\%$}
  & \dualb{0.050}{$-88\%$}
  & \avg{$-49\%$} \\
\midrule
\multicolumn{17}{@{}l@{}}{\emph{MeanDiff}} \\
MeanDiff-SVM (Eucl.)
  & \dual{25.75}{$+0\%$}
  & \dual{0.247}{$+2\%$}
  & \dual{0.425}{$+0\%$}
  & \avg{$+1\%$}
  & \dual{0.86}{$-90\%$}
  & \dual{0.000}{--}
  & \dual{0.600}{$+0\%$}
  & \avg{$-45\%$}
  & \dual{11.05}{$-59\%$}
  & \dual{0.000}{--}
  & \dual{0.625}{$+0\%$}
  & \avg{$-29\%$}
  & \dual{5.24}{$-64\%$}
  & \dual{0.368}{$+0\%$}
  & \dual{0.375}{$-6\%$}
  & \avg{$-23\%$} \\
MeanDiff-SVM (sph.)
  & \dual{17.73}{$-31\%$}
  & \dual{0.248}{$+2\%$}
  & \dual{0.425}{$+0\%$}
  & \avg{$-10\%$}
  & \dual{2.76}{$-68\%$}
  & \dual{0.000}{--}
  & \dual{0.625}{$+4\%$}
  & \avg{$-32\%$}
  & \dual{13.67}{$-49\%$}
  & \dual{0.000}{--}
  & \dual{0.625}{$+0\%$}
  & \avg{$-25\%$}
  & \dual{14.02}{$-4\%$}
  & \dual{0.412}{$+12\%$}
  & \dual{0.375}{$-6\%$}
  & \avg{$+0\%$} \\
MeanDiff-LR (Eucl.)
  & \dual{25.75}{$+0\%$}
  & \dual{0.247}{$+2\%$}
  & \dual{0.425}{$+0\%$}
  & \avg{$+1\%$}
  & \dual{3.09}{$-64\%$}
  & \dual{0.000}{--}
  & \dual{0.600}{$+0\%$}
  & \avg{$-32\%$}
  & \dual{12.14}{$-55\%$}
  & \dual{0.000}{--}
  & \dual{0.625}{$+0\%$}
  & \avg{$-27\%$}
  & \dual{5.77}{$-61\%$}
  & \dual{0.368}{$+0\%$}
  & \dual{0.375}{$-6\%$}
  & \avg{$-22\%$} \\
MeanDiff-LR (sph.)
  & \dual{25.75}{$+0\%$}
  & \dual{0.239}{$-2\%$}
  & \dual{0.425}{$+0\%$}
  & \avg{$-1\%$}
  & \dual{6.15}{$-29\%$}
  & \dual{0.000}{--}
  & \dual{0.625}{$+4\%$}
  & \avg{$-13\%$}
  & \dual{21.21}{$-21\%$}
  & \dual{0.000}{--}
  & \dual{0.625}{$+0\%$}
  & \avg{$-11\%$}
  & \dual{13.56}{$-7\%$}
  & \dual{0.400}{$+9\%$}
  & \dual{0.400}{$+0\%$}
  & \avg{$+0\%$} \\
\midrule
\multicolumn{17}{@{}l@{}}{\emph{BendVLM \cite{gerych2024bendvlm}}} \\
BendVLM (Eucl.)
  & \dual{13.03}{$-49\%$}
  & \dual{0.267}{$+10\%$}
  & \dual{0.125}{$-71\%$}
  & \avg{$-37\%$}
  & \dual{2.71}{$-69\%$}
  & \dual{0.000}{--}
  & \dual{0.625}{$+4\%$}
  & \avg{$-32\%$}
  & \dual{10.81}{$-60\%$}
  & \dual{0.000}{--}
  & \dual{0.625}{$+0\%$}
  & \avg{$-30\%$}
  & \dual{5.31}{$-64\%$}
  & \dual{0.374}{$+2\%$}
  & \dual{0.350}{$-13\%$}
  & \avg{$-25\%$} \\
BendVLM (sph.)
  & \dual{112.03}{$+335\%$}
  & \dualb{0.000}{$-100\%$}
  & \dualb{0.000}{$-100\%$}
  & \avg{$+45\%$}
  & \dualb{0.00}{$-100\%$}
  & \dual{0.452}{--}
  & \invalidrun
  & \avg{$-100\%$}\rlap{$^\dagger$}
  & \dualb{2.83}{$-89\%$}
  & \dual{0.000}{--}
  & \dual{0.550}{$-12\%$}
  & \avg{$-51\%$}
  & \dualb{2.86}{$-80\%$}
  & \dualb{0.152}{$-59\%$}
  & \dual{0.425}{$+6\%$}
  & \avg{$-44\%$} \\
\midrule
\modelname\ (Ours)
  & \dual{16.85}{$-35\%$}
  & \dual{0.096}{$-61\%$}
  & \dual{0.125}{$-71\%$}
  & \avgb{$-55\%$}
  & \dual{1.36}{$-84\%$}
  & \dual{0.000}{--}
  & \dualb{0.025}{$-96\%$}
  & \avgb{$-90\%$}
  & \dual{8.48}{$-68\%$}
  & \dual{0.000}{--}
  & \dualb{0.050}{$-92\%$}
  & \avgb{$-80\%$}
  & \dual{5.10}{$-65\%$}
  & \dual{0.154}{$-58\%$}
  & \dual{0.175}{$-56\%$}
  & \avgb{$-60\%$} \\
\bottomrule[1.2pt]
\end{tabular}%
}
\caption{\textbf{Main results against external baselines.} Each (method, model) row uses a single best-avg-$\alpha$ selected from $\{0.25,0.5,0.75,1.0,1.5\}$. A gray ``--'' on the second line means the baseline is zero, so percent change is undefined. $^\ast$~marks runs where the steered model produced no parseable answers (BendVLM (sph.) collapses generation there, consistent with Table~\ref{tab:mmstar}; likely a difficulty of porting a CLIP-space method to generative activations, not a flaw of the original method). $^\dagger$~average over a single task, excluded from cross-method comparison. Avg\,$\Delta\%$ averages over non-zero-baseline tasks: three on Pixtral-12B and Qwen3-VL-4B, two on the LLaVA models, so it is not strictly comparable across models.}
\label{tab:main_results}
\end{table*}

\begin{table*}[!t]
\centering
\setlength{\tabcolsep}{2.5pt}
\renewcommand{\arraystretch}{1.05}
\scriptsize
\providecommand{\mmdual}[2]{\begin{tabular}{@{}r@{}}#1\\[-1.5pt]{\tiny\color{gray!55!black}#2}\end{tabular}}
\providecommand{\mmdualb}[2]{\begin{tabular}{@{}r@{}}\textbf{#1}\\[-1.5pt]{\tiny\color{gray!55!black}\textbf{#2}}\end{tabular}}
\resizebox{0.78\textwidth}{!}{%
\begin{tabular}{@{}lrrrrrrrrr@{}}
\toprule[1.2pt]
 & \multicolumn{4}{c}{\textbf{Race-task steering}} & \multicolumn{4}{c}{\textbf{Gender-task steering}} & \\
\cmidrule(lr){2-5}\cmidrule(lr){6-9}\cmidrule(l){10-10}
\textbf{Method} & \textbf{Pixtral} & \textbf{Vicuna} & \textbf{Mistral} & \textbf{Qwen3} & \textbf{Pixtral} & \textbf{Vicuna} & \textbf{Mistral} & \textbf{Qwen3} & \textbf{Avg.} \\
\midrule
Baseline & \mmdual{53.6}{$0.0\%$} & \mmdual{37.3}{$0.0\%$} & \mmdual{38.5}{$0.0\%$} & \mmdual{61.5}{$0.0\%$} & \mmdual{53.6}{$0.0\%$} & \mmdual{37.3}{$0.0\%$} & \mmdual{38.5}{$0.0\%$} & \mmdual{61.5}{$0.0\%$} & \mmdual{47.7}{$0.0\%$} \\
\midrule
INLP (Eucl.) & \mmdual{52.1}{$-1.5\%$} & \mmdual{37.2}{$-0.1\%$} & \mmdual{36.4}{$-2.1\%$} & \mmdual{62.1}{$+0.6\%$} & \mmdual{53.3}{$-0.3\%$} & \mmdual{37.6}{$+0.3\%$} & \mmdual{38.0}{$-0.5\%$} & \mmdual{62.0}{$+0.5\%$} & \mmdual{47.3}{$-0.4\%$} \\
INLP (sph.) & \mmdual{47.5}{$-6.1\%$} & \mmdual{36.8}{$-0.5\%$} & \mmdual{38.4}{$-0.1\%$} & \mmdualb{62.4}{$+0.9\%$} & \mmdual{50.3}{$-3.3\%$} & \mmdual{37.1}{$-0.2\%$} & \mmdual{38.6}{$+0.1\%$} & \mmdual{61.3}{$-0.2\%$} & \mmdual{46.5}{$-1.2\%$} \\
MeanDiff-SVM (Eucl.) & \mmdual{53.3}{$-0.3\%$} & \mmdual{37.1}{$-0.2\%$} & \mmdual{39.0}{$+0.5\%$} & \mmdual{61.4}{$-0.1\%$} & \mmdual{53.1}{$-0.5\%$} & \mmdual{37.4}{$+0.1\%$} & \mmdual{38.2}{$-0.3\%$} & \mmdual{61.9}{$+0.4\%$} & \mmdual{47.7}{$-0.1\%$} \\
MeanDiff-SVM (sph.) & \mmdual{53.2}{$-0.4\%$} & \mmdual{37.3}{$+0.0\%$} & \mmdual{39.2}{$+0.7\%$} & \mmdual{61.8}{$+0.3\%$} & \mmdual{53.1}{$-0.5\%$} & \mmdual{37.2}{$-0.1\%$} & \mmdual{39.0}{$+0.5\%$} & \mmdual{61.8}{$+0.3\%$} & \mmdualb{47.8}{$+0.1\%$} \\
MeanDiff-LR (Eucl.) & \mmdualb{53.6}{$+0.0\%$} & \mmdual{36.9}{$-0.4\%$} & \mmdual{38.7}{$+0.2\%$} & \mmdual{61.4}{$-0.1\%$} & \mmdualb{53.7}{$+0.1\%$} & \mmdual{37.2}{$-0.1\%$} & \mmdual{38.8}{$+0.3\%$} & \mmdual{61.8}{$+0.3\%$} & \mmdualb{47.8}{$+0.0\%$} \\
MeanDiff-LR (sph.) & \mmdual{53.0}{$-0.6\%$} & \mmdual{37.1}{$-0.2\%$} & \mmdual{39.4}{$+0.9\%$} & \mmdual{61.7}{$+0.2\%$} & \mmdual{53.2}{$-0.4\%$} & \mmdual{36.6}{$-0.7\%$} & \mmdual{38.8}{$+0.3\%$} & \mmdual{61.9}{$+0.4\%$} & \mmdual{47.7}{$+0.0\%$} \\
BendVLM (Eucl.) & \mmdual{53.0}{$-0.6\%$} & \mmdual{37.1}{$-0.2\%$} & \mmdual{39.0}{$+0.5\%$} & \mmdual{61.6}{$+0.1\%$} & \mmdual{45.1}{$-8.5\%$} & \mmdualb{37.8}{$+0.5\%$} & \mmdual{38.6}{$+0.1\%$} & \mmdual{61.7}{$+0.2\%$} & \mmdual{46.7}{$-1.0\%$} \\
BendVLM (sph.) & \mmdual{6.9}{$-46.7\%$} & \mmdual{2.9}{$-34.4\%$} & \mmdualb{39.5}{$+1.0\%$} & \mmdual{59.0}{$-2.5\%$} & \mmdual{5.3}{$-48.3\%$} & \mmdual{2.6}{$-34.7\%$} & \mmdualb{39.8}{$+1.3\%$} & \mmdual{59.3}{$-2.2\%$} & \mmdual{26.9}{$-20.8\%$} \\
\midrule
Pooled SVD (Eucl.) & \mmdual{47.3}{$-6.3\%$} & \mmdual{35.0}{$-2.3\%$} & \mmdual{38.9}{$+0.4\%$} & \mmdual{60.8}{$-0.7\%$} & \mmdual{46.5}{$-7.1\%$} & \mmdual{37.5}{$+0.2\%$} & \mmdual{39.0}{$+0.5\%$} & \mmdual{61.8}{$+0.3\%$} & \mmdual{45.9}{$-1.9\%$} \\
Pooled SVD (sph.) & \mmdual{53.5}{$-0.1\%$} & \mmdualb{38.2}{$+0.9\%$} & \mmdual{38.5}{$+0.0\%$} & \mmdual{61.7}{$+0.2\%$} & \mmdual{53.3}{$-0.3\%$} & \mmdual{37.2}{$-0.1\%$} & \mmdual{38.4}{$-0.1\%$} & \mmdual{61.8}{$+0.3\%$} & \mmdualb{47.8}{$+0.1\%$} \\
Per-token SVD (Eucl.) & \mmdual{42.0}{$-11.6\%$} & \mmdual{37.2}{$-0.1\%$} & \mmdual{37.7}{$-0.8\%$} & \mmdual{60.2}{$-1.3\%$} & \mmdual{22.1}{$-31.5\%$} & \mmdual{36.9}{$-0.4\%$} & \mmdual{38.5}{$+0.0\%$} & \mmdual{61.0}{$-0.5\%$} & \mmdual{42.0}{$-5.8\%$} \\
Per-token SVD (sph.) & \mmdual{50.1}{$-3.5\%$} & \mmdual{36.2}{$-1.1\%$} & \mmdual{38.1}{$-0.4\%$} & \mmdual{61.1}{$-0.4\%$} & \mmdual{50.9}{$-2.7\%$} & \mmdual{37.4}{$+0.1\%$} & \mmdual{38.2}{$-0.3\%$} & \mmdual{61.2}{$-0.3\%$} & \mmdual{46.6}{$-1.1\%$} \\
\midrule
\textbf{\modelname\ (ours)} & \mmdual{\textbf{53.2}}{\textbf{$-0.4\%$}} & \mmdual{\textbf{37.8}}{\textbf{$+0.5\%$}} & \mmdual{\textbf{38.9}}{\textbf{$+0.4\%$}} & \mmdual{\textbf{61.6}}{\textbf{$+0.1\%$}} & \mmdual{\textbf{53.3}}{\textbf{$-0.3\%$}} & \mmdual{\textbf{37.1}}{\textbf{$-0.2\%$}} & \mmdual{\textbf{38.4}}{\textbf{$-0.1\%$}} & \mmdualb{62.1}{$+0.6\%$} & \mmdualb{47.8}{$+0.1\%$} \\
\bottomrule
\end{tabular}
}
\caption{\textbf{MMStar capability preservation at best-avg-$\alpha$.} Best-avg-$\alpha$ choices are listed in Appendix~\ref{app:hparams}.}
\label{tab:mmstar}
\vspace{-0.4em}
\end{table*}

\begin{table}[!ht]
\centering
\setlength{\tabcolsep}{4.0pt}
\renewcommand{\arraystretch}{1.05}
\scriptsize

\begin{tabular}{@{}l r r@{}}
\toprule[1.2pt]
\multicolumn{3}{c}{\textbf{Component ablation on Qwen3-VL-4B, MCQ (race)}} \\
\midrule
\textbf{Variant} & JSD$\times$1k $\downarrow$ & Bias Red.\ \% $\uparrow$ \\
\midrule
Baseline (unsteered)                                & 14.646       & --              \\
Pooled SVD (Eucl.; hard proj., no gate)             & 4.261        & 70.9            \\
Pooled SVD (sph.; hard proj., no gate)              & 5.612        & 61.7            \\
\modelname\ w/o gate ($g_i\!\equiv\!1$)              & 5.612        & 61.7            \\
\modelname\ w/o Slerp (hard proj.\ + gate)           & 4.533        & 69.1            \\
\modelname\ \textbf{(full)}                         & \best{2.414} & \best{83.5}     \\
\bottomrule[1.2pt]
\end{tabular}
\caption{\textbf{Component ablation of \modelname} on Qwen3-VL-4B, MCQ salary (race) at $\alpha=1.0$. }
\label{tab:component_ablation}
\end{table}

\begin{table}[!ht]
\centering
\setlength{\tabcolsep}{8.0pt}
\renewcommand{\arraystretch}{1.05}
\scriptsize
\begin{tabular}{@{}l c c c c@{}}
\toprule[1.2pt]
\multicolumn{5}{c}{\textbf{Sensitivity to $(\kappa, g_{\mathrm{floor}})$ on Qwen3-VL-4B, MCQ}} \\
\midrule
$g_{\mathrm{floor}}\!\setminus\!\kappa$ & $1$ & $2$ & $5$ & $10$ \\
\midrule
$0.0$  & 4.851 &  4.578 & 10.971 &  9.196 \\
$0.3$  & 4.578 &  5.781 & \best{2.414} &  4.533 \\
$0.5$  & 6.128 &  6.714 &  4.578 &  4.069 \\
\bottomrule[1.2pt]
\end{tabular}
\caption{\textbf{Sensitivity of \modelname} to gate sharpness $\kappa$ and floor $g_{\mathrm{floor}}$ on Qwen3-VL-4B, MCQ salary (race) at $\alpha=1.0$. Entries are JSD$\times$1k; lower is better.}
\label{tab:sensitivity_ablation}
\end{table}

\subsection{Setup}
\label{sec:exp_setup}

\paragraph{Models}
We evaluate \modelname\ on four generative VLMs spanning three model families and a range of parameter scales: Pixtral-12B (12B) \cite{agrawal2024pixtral}, LLaVA-1.6-Vicuna-7B (7B), LLaVA-1.6-Mistral-7B (7B)~\cite{liu2024llavanext}, and Qwen3-VL-4B-Instruct (4B)~\cite{bai2025qwen3}. For every method, we intervene at the same late vision-to-language projection layer so that differences reflect the steering rule rather than the hook location. Checkpoint identifiers, exact hook paths, activations, and results for additional hook sites appear in Appendix~\ref{app:layers}.

\paragraph{Counterfactual images}
The counterfactual image sets are not generated by us: we use the $480$ real-photograph, face-only counterfactual images released with the REFLECT/FOCUS dataset \cite{chen2026reflect}, covering $6$ occupations $\times$ $8$ source identities $\times$ $5$ perceived races $\times$ $2$ genders. Each identity's variants are pixel-aligned, so activation differences isolate the protected attribute; discovery occupations are disjoint from the evaluated occupation (\emph{Protocol} below).

\paragraph{Tasks and metrics}
We evaluate with three bias protocols on categorical and pairwise attributes. For each counterfactual image, the probes ask:
\begin{itemize}[nolistsep,left=8pt]
    \item \textbf{MCQ} (race, multiple choice; REFLECT \cite{chen2026reflect}): a salary/education multiple-choice question about the pictured person, scored by the Jensen--Shannon divergence of answer distributions across races. The metric is \texttt{mean\_race\_jsd}$\times 10^3$.
    \item \textbf{2AFC} (race, pairwise; REFLECT \cite{chen2026reflect}): a two-alternative forced choice (``which of two versions of the same person has higher income'') on race-paired images. The metric is \texttt{race\_bias\_std}.
    \item \textbf{Nurse/Doctor} (gender, classification): a single-word occupation query on nurse/doctor images. The metric is the gender gap in correct classification, $|P(\text{nurse}|\text{man})-P(\text{nurse}|\text{woman})|$.
\end{itemize}
Two of the four models (LLaVA-1.6-Vicuna-7B and LLaVA-1.6-Mistral-7B) exhibit zero baseline 2AFC bias. Steered values are reported for completeness but do not contribute to the average. We also report accuracy on MMStar \cite{chen2024we}, a 1{,}500-question multiple-choice benchmark of general multimodal capability spanning six dimensions from coarse/fine-grained perception to reasoning, math, and science, scored by exact matching with VLMEvalKit \cite{duan2024vlmevalkit}; per-dimension breakdown in Appendix~\ref{app:mmstar_dims}.

\paragraph{Baselines}
This work \emph{re-implements and adapts} existing inference-time debiasers from CLIP- and LM-era settings to generative VLMs. These methods were originally developed for static word embeddings, BERT-style encoders, or CLIP dual-encoders; here, we make them target the same late vision-to-language projection layer for a controlled comparison. Concretely, we compare \modelname\ against ten external steering baselines spanning four families:
\begin{itemize}[nolistsep,left=8pt]
    \item \textbf{INLP} \cite{ravfogel2020null}: iterative linear-probe null-space stacking, ported from word-embedding debiasing, with Euclidean and spherical variants.
    \item \textbf{MeanDiff}: classifier-guided mean-shift correction in the style of \citet{zou2023representation} probing. We use both RBF-SVM and multinomial logistic-regression probes, each with Euclidean and spherical variants trained on pooled visual activations.
    \item \textbf{BendVLM} \cite{gerych2024bendvlm}: we adapt BendVLM's CLIP-embedding retrieval and Lagrangian equalization from the shared text--image embedding space to generative-VLM hidden activations, evaluating both Euclidean and spherical variants.
    \item \textbf{LEACE} \cite{belrose2023leace}: the closed-form least-squares concept eraser, fit with the authors' official implementation, as a direct spherical port and combined with our calibrated gate (Appendix~\ref{app:baselines_impl}).
\end{itemize}
We additionally evaluate a \emph{prompt-space} family: the identical probes with a fairness instruction appended to every prompt, with no steering (\S\ref{sec:added_baselines}). All adapted implementations are released together with this paper as a benchmark suite for inference-time debiasing of generative VLMs.

\paragraph{Protocol}
For every (method, model) pair we sweep $\alpha\in\{0.25,0.5,0.75,1.0,1.5\}$ across the bias tasks and pick a \emph{single} $\alpha$ that minimizes the unweighted mean per-task \% change relative to the unsteered baseline (``best-avg-$\alpha$'' protocol). Only tasks with non-zero baselines contribute to the average. All bias metrics, the Avg\,$\Delta\%$ column, and MMStar for that method and model are reported at this selected operating point. For \modelname\ we fix $\kappa=5$ and $g_{\mathrm{floor}}=0.3$ throughout, with sensitivity analyses in \S\ref{sec:ablation} and Appendix~\ref{app:hparams}. Discovery uses counterfactual image sets whose identities are held out from the evaluation occupation. For the race-based tasks we discover on $\{\texttt{cook},\texttt{doctor},\texttt{lawyer},\texttt{nurse},\texttt{teacher}\}$, and for the gender-based Nurse/Doctor task we discover on $\{\texttt{cook},\texttt{lawyer},\texttt{teacher}\}$ while holding out both \texttt{nurse} and \texttt{doctor}. MCQ and 2AFC are scored on the SocialCounterfactuals probe suite \cite{howard2024socialcounterfactuals}, which does not share identities with the discovery pool.

\paragraph{Reading the tables}
Tables~\ref{tab:main_results} and~\ref{tab:mmstar} should be read as operating-point comparisons: for each method and model, the same selected $\alpha$ is used for all bias metrics and for MMStar. They separate two questions that are often conflated: whether a steering rule reduces demographic sensitivity, and whether the same operating point preserves general multimodal reasoning. A useful method must do both.

\subsection{Main Results}
\label{sec:main_results}

Table~\ref{tab:main_results} reports best-avg-$\alpha$ metrics for every method across all models and bias tasks; \modelname\ is the \emph{single most reliable} method in the comparison.

\paragraph{Ranking}
\modelname\ attains the lowest Avg\,$\Delta\%$ among all external baselines on all $4$ models. The gaps are $-55\%$ vs.\ $-37\%$ on Pixtral-12B, where BendVLM (Eucl.) is the strongest external baseline, then $-90\%$ vs.\ $-86\%$ on LLaVA-1.6-Vicuna-7B, $-80\%$ vs.\ $-78\%$ on LLaVA-1.6-Mistral-7B, and $-60\%$ vs.\ $-49\%$ on Qwen3-VL-4B, all three against INLP (sph.), the strongest competitor overall.

\paragraph{Magnitude}
\modelname\ produces substantial bias reductions: up to $-96\%$ on Nurse/Doctor (LLaVA-1.6-Vicuna-7B: $0.600\rightarrow 0.025$), $-84\%$ on MCQ (LLaVA-1.6-Vicuna-7B: $8.70\rightarrow 1.36$), and $-61\%$ on 2AFC (Pixtral-12B: $0.243\rightarrow 0.096$).

MCQ and Nurse/Doctor measure different output formats---categorical answer distributions and binary occupation classification---yet \modelname\ reduces both without changing the selected operating point for a given model; on the two LLaVA models the zero-baseline 2AFC entries serve as sanity checks.

\paragraph{Capability preservation (MMStar)}
Table~\ref{tab:mmstar} reports MMStar results at each method's best-avg-$\alpha$ (the same $\alpha$ used to report its bias metrics in Table~\ref{tab:main_results}) across all steering settings. \modelname\ stays within $\pm 0.6$\,p.p.\ of baseline across these evaluations: adaptive geodesic steering preserves general capability. A steering method that reduces benchmark bias by collapsing general VLM capability is not useful as a deployment primitive, and several baselines suffer exactly such drops; BendVLM (sph.), for example, collapses on Pixtral and LLaVA-Vicuna. \modelname's MMStar changes remain small under both race and gender steering: the gate confines the intervention to directions that carry the targeted bias and leaves the remaining visual tokens largely untouched.

\begin{table}[!ht]
\centering
\setlength{\tabcolsep}{4.0pt}
\renewcommand{\arraystretch}{1.05}
\scriptsize
\begin{tabular}{@{}l c c c@{}}
\toprule[1.2pt]
\textbf{Model} & \makecell{\modelname\ vs.\\unsteered ($p$)} & \makecell{MMStar $\Delta$pp\\(race / gender)} & \makecell{McNemar $p$\\(race / gender)} \\
\midrule
Pixtral-12B         & $<\!10^{-4}$                                & $-0.40$ / $-0.33$ & $0.58$ / $0.49$ \\
LLaVA-Vicuna-7B     & $0.068^{\,a}$                               & $+0.47$ / $-0.20$ & $0.58$ / $0.77$ \\
LLaVA-Mistral-7B    & $0.007$                                     & $+0.40$ / $-0.07$ & $0.51$ / $1.00$ \\
Qwen3-VL-4B         & $0.002$                                     & $+0.13$ / $+0.60$ & $0.86$ / $0.09$ \\
\bottomrule[1.2pt]
\end{tabular}
\caption{\textbf{Significance summary at the paper's operating points.} Column 2: paired sign-flip permutation test of \modelname\ vs.\ the unsteered model, combined across non-zero-baseline bias tasks ($^{a}$on LLaVA-Vicuna the Nurse/Doctor task alone has $p<10^{-4}$). Columns 3--4: MMStar paired per-question statistics under race / gender steering. Protocol details in Appendix~\ref{app:significance}.}
\label{tab:significance}
\end{table}

\begin{table}[!ht]
\centering
\setlength{\tabcolsep}{3.0pt}
\renewcommand{\arraystretch}{1.05}
\scriptsize
\begin{tabular}{@{}l r r r r r r@{}}
\toprule[1.2pt]
\textbf{Method} & \makecell{Pix-\\tral} & \makecell{Qwen3-\\VL} & \makecell{Mis-\\tral} & \makecell{Vi-\\cuna} & mean & worst \\
\midrule
\modelname\ (paper)      & $-55.3$ & \best{$-59.9$} & \best{$-80.2$} & \best{$-90.1$} & \best{$-71.4$} & \best{$-55.3$} \\
LEACE + our gate         & \best{$-72.2$} & $-30.8$ & $-66.9$ & $-83.9$ & $-63.5$ & $-30.8$ \\
INLP (sph.)              & $-31.2$ & $-49.3$ & $-77.6$ & $-86.1$ & $-61.1$ & $-31.2$ \\
LEACE (sph.)             & $-70.6$ & $-26.4$ & $-51.5$ & $-45.4$ & $-48.5$ & $-26.4$ \\
\bottomrule[1.2pt]
\end{tabular}
\caption{\textbf{Added baselines.} Avg\,$\Delta\%$ at each method's best-avg-$\alpha$ (bias tasks with non-zero baseline; lower is better), with the four-model mean and worst case. INLP (sph.) is the strongest baseline from Table~\ref{tab:main_results}, shown for reference. Full details in Appendix~\ref{app:leace_results}.}
\label{tab:added_baselines}
\end{table}

\paragraph{Statistical significance}
Decoding is greedy and deterministic, so the relevant uncertainty is sampling over evaluation items. For every reported cell we compute 95\% bootstrap CIs (10{,}000 resamples, over items and identities), paired sign-flip permutation tests, and exact McNemar tests on MMStar (Appendix~\ref{app:significance}). Table~\ref{tab:significance} summarizes: \modelname's reductions are significant on three of four backbones (on LLaVA-Vicuna the Nurse/Doctor reduction alone has $p<10^{-4}$), and its MMStar changes are statistically indistinguishable from the unsteered model in all eight runs, whereas the strongest competitor, INLP (sph.), significantly damages Pixtral's MMStar ($-6.1$\,pp race, $p<10^{-4}$).

\paragraph{Held-out $\alpha$ selection}
Table~\ref{tab:main_results} selects $\alpha$ on the probe suite it reports, identically for every method, so no method is advantaged, but absolute reductions may be optimistic. Under a cross-fitted protocol selecting $\alpha$ on one identity fold and evaluating on the other (Appendix~\ref{app:heldout_alpha}), \modelname\ stays near the top on every model, matches the paper's $\alpha$ on six of eight folds, and reduces bias in every held-out fold.

\paragraph{Additional baseline families}
\label{sec:added_baselines}
Table~\ref{tab:added_baselines} adds LEACE \cite{belrose2023leace} at the same layer under the same protocol. Ported this way, LEACE is a serious baseline (it wins the Pixtral column while preserving MMStar there) but is less consistent across backbones ($-26\%$ to $-31\%$ on Qwen3-VL), while \modelname\ keeps the best four-model mean and worst case. Prompt-space mitigation (a fairness instruction appended to every prompt) is inconsistent and can even amplify bias (e.g., Vicuna MCQ $+85\%$); full results in Appendix~\ref{app:prompt_baseline}.

\paragraph{Debiasing, not blanket removal}
A steering rule could reduce measured bias by deleting demographic perception outright; \modelname\ does not. Steering one attribute leaves recognition of the \emph{other} attribute intact on all four backbones, all six MMStar dimensions stay within binomial noise, and MME (Qwen3-VL) moves only $-0.4\%$ / $-1.8\%$ with no subcategory collapsing. The trade-off with the \emph{targeted} attribute is controlled by $\alpha$: at $\alpha=0.5$ on Qwen3-VL, race recognition stays within $4$\,pp of baseline while $55\%$ of the MCQ reduction is already realized, so tasks that need the attribute can run at moderate $\alpha$ or gate steering off (Appendix~\ref{app:recognition}).

\subsection{Ablations and Sensitivity}
\label{sec:ablation}

\paragraph{Component and sensitivity}
Table~\ref{tab:component_ablation} isolates the contribution of each \modelname\ ingredient on a representative setting (Qwen3-VL-4B, MCQ salary, $\alpha\!=\!1.0$). Disabling the gate ($g_i\!\equiv\!1$) gives the same JSD as the ungated spherical hard-projection variant (5.61). Adding the gate \emph{without} Slerp improves this value from 5.61 to 4.53 ($-69\%$ from baseline), showing that selective gating alone already helps. The full method, combining gating \emph{and} Slerp interpolation, achieves 2.41 ($-84\%$), a further 14\,pp improvement in this setting.

Table~\ref{tab:sensitivity_ablation} sweeps gate sharpness $\kappa\!\in\!\{1,2,5,10\}$ and floor $g_{\mathrm{floor}}\!\in\!\{0.0,0.3,0.5\}$. All twelve settings reduce bias well below the unsteered baseline (14.65); the best reduction ($-84\%$) occurs at $\kappa\!=\!5,\;g_{\mathrm{floor}}\!=\!0.3$. High $\kappa$ with $g_{\mathrm{floor}}\!=\!0$ can over-select, but a positive floor recovers performance.

\paragraph{Geometry isolated at matched gate}
Toggling Euclidean versus spherical steering at matched gate condition on all four backbones, with paired sign-flip tests (Appendix~\ref{app:geometry_ablation}), shows the toggle is within sampling noise at the 4B--7B scales but decisive at 12B: on Pixtral the spherical form is significantly better under both toggles ($p=0.031$ / $0.015$), both Euclidean forms significantly damage MMStar there ($-6.3$\,pp, $p<10^{-8}$) while the spherical forms preserve it everywhere, and over-steering exposes failure modes the geodesic never showed ($9\times$ overshoot; a degenerate cell with $72/80$ unparseable). The ablations thus attribute \emph{robustness that grows with model scale} to on-sphere steering, and \emph{selectivity and controllability} to the gate.

\section{Conclusion}
\label{sec:conclusion}

We introduced \modelname, an inference-time debiasing framework for generative VLMs that discovers counterfactual spherical bias subspaces, applies a calibrated token gate, and steers with norm-preserving Slerp. Across four generative VLMs and three bias protocols, \modelname\ attains the lowest Avg\,$\Delta\%$ on all $4$ models against ten external steering baselines and prompt-based mitigation, significant on three of four backbones, while keeping MMStar indistinguishable from the unsteered model. Matched-gate ablations attribute robustness at scale to on-sphere steering and controllable selectivity to the gate, supporting adaptive geodesic steering as a practical primitive for inference-time bias mitigation.

\clearpage

\section*{Limitations}

\paragraph{Scope}
Our evaluation covers four generative VLMs, three demographic-bias protocols, and MMStar, but does not establish that \modelname\ generalizes to all architectures, languages, domains, or bias types. We focus on perceived race and gender in image-conditioned settings; other protected attributes, intersectional groups, multilingual prompts, and open-ended uses remain future work. The subspace construction may also extend to non-demographic bias axes such as political leaning, whose embedding-space structure has been studied in text models \cite{sun2025prism}.

\paragraph{Calibration}
\modelname\ still requires selecting a steering strength. We choose one best-avg-$\alpha$ per (method, model) from $\{0.25,0.5,0.75,1.0,1.5\}$ using held-out bias measurements, and the optimal value is model-dependent. A tuning-free rule, potentially using the bias-norm statistics $\tilde{b}$ and $\sigma_b$ to set token-level strengths, remains an open direction.

\paragraph{Deployment}
Lower benchmark bias should not be interpreted as complete fairness or safety. \modelname\ changes activations at inference time, but does not remove biased knowledge from model parameters or guarantee robustness under distribution shift. Misconfigured steering may also suppress useful demographic information or affect tasks beyond our capability checks: at the bias-minimizing steering strength, the targeted attribute's reportability on face-only close-ups is strongly reduced (Appendix~\ref{app:recognition}), so tasks that legitimately require the attribute, such as requested demographic descriptions or clinical documentation, should operate at a moderate $\alpha$ chosen from the trade-off curve, or gate steering off. Deployment should include broader auditing, human review where appropriate, and monitoring for residual or newly introduced harms.

\section*{Ethics Statement}

\paragraph{Intended use and dual use}
\modelname\ is designed to reduce demographic bias in deployed VLMs without retraining. The same mechanism, however, is dual-use: a steering subspace can be aimed at amplifying demographic sensitivity as easily as at reducing it, and inference-time manipulation of model behavior is demonstrably exploitable \cite{liu2026reasoning}. At high steering strength the intervention also shades from \emph{debiasing} into \emph{attribute removal} (Appendix~\ref{app:recognition}). We release the trade-off measurements needed to choose an operating point deliberately; deployments should measure both bias and attribute retention on the target task before fixing a steering strength.

\paragraph{Demographic labels}
The counterfactual datasets we use annotate \emph{perceived} race and binary gender as assigned by dataset curators. These labels are operational constructs for measuring output disparities; they do not capture self-identification, intersectional identity, or the full diversity of human appearance, and results should not be read as claims about any individual's identity.

\paragraph{Data and licensing}
All images come from publicly released research datasets (REFLECT/FOCUS, SocialCounterfactuals) used within their stated licenses; we generate no new images of people and release no new personal data. Artifact documentation and intended-use notes accompany the released benchmark suite (Appendix~\ref{sec:artifact_statement}).

\paragraph{Residual risk}
Reduced benchmark bias is not fairness. Steering does not remove biased knowledge from model parameters, may behave differently under distribution shift, and covers only the attributes and tasks we measure. Systems using \modelname\ in consequential settings should retain human oversight and independent auditing.

\bibliography{reference}

\begin{thebibliography}{55}
\providecommand{\natexlab}[1]{#1}

\bibitem[{Agrawal et~al.(2024)Agrawal, Antoniak, Hanna, Bout, Chaplot,
  Chudnovsky, Costa, Monicault, Garg, Gervet, Ghosh, H\'{e}liou, Jacob, Jiang,
  Khandelwal, Lacroix, Lample, Casas, Lavril, Scao, Lo, Marshall, Martin,
  Mensch, Muddireddy, Nemychnikova, Pellat, Platen, Raghuraman, Rozi\'{e}re,
  Sablayrolles, Saulnier, Sauvestre, Shang, Soletskyi, Stewart, Stock, Studnia,
  Subramanian, Vaze, Wang, and Yang}]{agrawal2024pixtral}
Pravesh Agrawal, Szymon Antoniak, Emma~Bou Hanna, Baptiste Bout, Devendra
  Chaplot, Jessica Chudnovsky, Diogo Costa, Baudouin~De Monicault, Saurabh
  Garg, Theophile Gervet, Soham Ghosh, Am\'{e}lie H\'{e}liou, Paul Jacob,
  Albert~Q. Jiang, Kartik Khandelwal, Timoth\'{e}e Lacroix, Guillaume Lample,
  Diego~Las Casas, Thibaut Lavril, and 23 others. 2024.
\newblock Pixtral 12b.
\newblock \emph{arXiv preprint arXiv:2410.07073}.

\bibitem[{An et~al.(2026)An, Jang, Hirota, Hachiuma, Augenstein, and
  Shim}]{an2026interpretable}
Na~Min An, Yoonna Jang, Yusuke Hirota, Ryo Hachiuma, Isabelle Augenstein, and
  Hyunjung Shim. 2026.
\newblock Interpretable debiasing of vision-language models for social
  fairness.
\newblock \emph{arXiv preprint arXiv:2602.24014}.

\bibitem[{Bai et~al.(2025)Bai, Cai, Chen, Chen, Chen, Cheng, Deng, Ding, Gao,
  Ge, Ge, Guo, Huang, Huang, Huang, Hui, Jiang, Li, Li, Li, Li, Lin, Lin, Liu,
  Liu, Liu, Liu, Liu, Liu, Lu, Luo, Lv, Men, Meng, Ren, Ren, Song, Sun, Tang,
  Tu, Wan, Wang, Wang, Wang, Wang, Xie, Xu, Xu, Xu, Yang, Yang, Yang, Yang, Yu,
  Zhang, Zhang, Zhang, Zheng, Zhong, Zhou, Zhou, Zhou, Zhu, and
  Zhu}]{bai2025qwen3}
Shuai Bai, Yuxuan Cai, Ruizhe Chen, Keqin Chen, Xionghui Chen, Zesen Cheng,
  Lianghao Deng, Wei Ding, Chang Gao, Chunjiang Ge, Wenbin Ge, Zhifang Guo,
  Qidong Huang, Jie Huang, Fei Huang, Binyuan Hui, Shutong Jiang, Zhaohai Li,
  Mingsheng Li, and 45 others. 2025.
\newblock Qwen3-vl technical report.
\newblock \emph{arXiv preprint arXiv:2511.21631}.

\bibitem[{Banerjee et~al.(2024)Banerjee, Teso, Sayin, and
  Passerini}]{banerjee2024learning}
Debodeep Banerjee, Stefano Teso, Burcu Sayin, and Andrea Passerini. 2024.
\newblock Learning to guide human decision makers with vision-language models.
\newblock \emph{arXiv preprint arXiv:2403.16501}.

\bibitem[{Belrose et~al.(2023)Belrose, Schneider-Joseph, Ravfogel, Cotterell,
  Raff, and Biderman}]{belrose2023leace}
Nora Belrose, David Schneider-Joseph, Shauli Ravfogel, Ryan Cotterell, Edward
  Raff, and Stella Biderman. 2023.
\newblock {LEACE}: Perfect linear concept erasure in closed form.
\newblock \emph{Advances in Neural Information Processing Systems},
  36:66044--66063.

\bibitem[{Chen et~al.(2025)Chen, Li, Fan, Chen, Zhou, Wang, and
  Li}]{chen2025growth}
Die Chen, Zhiwen Li, Mingyuan Fan, Cen Chen, Wenmeng Zhou, Yanhao Wang, and
  Yaliang Li. 2025.
\newblock Growth inhibitors for suppressing inappropriate image concepts in
  diffusion models.
\newblock In \emph{International Conference on Learning Representations},
  volume 2025, pages 79164--79184.

\bibitem[{Chen et~al.(2026)Chen, Huang, Zhao, Jiang, Chang, and
  Yu}]{chen2026reflect}
Haodong Chen, Qiang Huang, Jiaqi Zhao, Qiuping Jiang, Xiaojun Chang, and Jun
  Yu. 2026.
\newblock Measuring social bias in vision-language models with face-only
  counterfactuals from real photos.
\newblock \emph{arXiv preprint arXiv:2601.06931}.

\bibitem[{Chen et~al.(2024)Chen, Li, Dong, Zhang, Zang, Chen, Duan, Wang, Qiao,
  Lin, and Zhao}]{chen2024we}
Lin Chen, Jinsong Li, Xiaoyi Dong, Pan Zhang, Yuhang Zang, Zehui Chen, Haodong
  Duan, Jiaqi Wang, Yu~Qiao, Dahua Lin, and Feng Zhao. 2024.
\newblock Are we on the right way for evaluating large vision-language models?
\newblock \emph{arXiv preprint arXiv:2403.20330}.

\bibitem[{Cheng et~al.(2025)Cheng, Guo, Guo, Yang, Gan, Guan, and
  Nie}]{cheng2025social}
Harry Cheng, Yangyang Guo, Qingpei Guo, Ming Yang, Tian Gan, Weili Guan, and
  Liqiang Nie. 2025.
\newblock Social debiasing for fair multi-modal llms.
\newblock In \emph{Proceedings of the IEEE/CVF International Conference on
  Computer Vision}, pages 1740--1750.

\bibitem[{Chuang et~al.(2023)Chuang, Jampani, Li, Torralba, and
  Jegelka}]{chuang2023biasedprompts}
Ching-Yao Chuang, Varun Jampani, Yuanzhen Li, Antonio Torralba, and Stefanie
  Jegelka. 2023.
\newblock Debiasing vision-language models via biased prompts.
\newblock \emph{arXiv preprint arXiv:2302.00070}.

\bibitem[{Duan et~al.(2024)Duan, Yang, Qiao, Fang, Chen, Liu, Dong, Zang,
  Zhang, Wang, Lin, and Chen}]{duan2024vlmevalkit}
Haodong Duan, Junming Yang, Yuxuan Qiao, Xinyu Fang, Lin Chen, Yuan Liu, Xiaoyi
  Dong, Yuhang Zang, Pan Zhang, Jiaqi Wang, Dahua Lin, and Kai Chen. 2024.
\newblock \href {https://doi.org/10.1145/3664647.3685520} {Vlmevalkit: An
  open-source toolkit for evaluating large multi-modality models}.
\newblock In \emph{Proceedings of the 32nd ACM International Conference on
  Multimedia}, MM '24, page 11198–11201.

\bibitem[{Feng et~al.(2025)Feng, Sun, Sun, Zhu, Huang, Tung, and
  Chen}]{feng2025don}
Yingchaojie Feng, Yiqun Sun, Yandong Sun, Minfeng Zhu, Qiang Huang, Anthony
  Kum~Hoe Tung, and Wei Chen. 2025.
\newblock Don't reinvent the wheel: Efficient instruction-following text
  embedding based on guided space transformation.
\newblock In \emph{Proceedings of the 63rd Annual Meeting of the Association
  for Computational Linguistics (Volume 1: Long Papers)}, pages 24511--24525.

\bibitem[{Gerych et~al.(2026)Gerych, Parent, Perian, Javed, Solomon, and
  Ghassemi}]{wring2026}
Walter Gerych, Cassandra Parent, Quinn Perian, Rafiya Javed, Justin Solomon,
  and Marzyeh Ghassemi. 2026.
\newblock Wring out the bias: A rotation-based alternative to projection
  debiasing.
\newblock In \emph{The Fourteenth International Conference on Learning
  Representations}.

\bibitem[{Gerych et~al.(2024)Gerych, Zhang, Hamidieh, Pan, Sharma, Hartvigsen,
  and Ghassemi}]{gerych2024bendvlm}
Walter Gerych, Haoran Zhang, Kimia Hamidieh, Eileen Pan, Maanas~K Sharma, Tom
  Hartvigsen, and Marzyeh Ghassemi. 2024.
\newblock Bendvlm: Test-time debiasing of vision-language embeddings.
\newblock \emph{Advances in Neural Information Processing Systems},
  37:62480--62502.

\bibitem[{Ghate et~al.(2025)Ghate, Charlesworth, Diab, and
  Caliskan}]{ghate2025propagate}
Kshitish Ghate, Tessa Charlesworth, Mona Diab, and Aylin Caliskan. 2025.
\newblock Biases propagate in encoder-based vision-language models: A
  systematic analysis from intrinsic measures to zero-shot retrieval outcomes.
\newblock In \emph{Findings of the Association for Computational Linguistics:
  ACL 2025}, pages 18562--18580.

\bibitem[{Hall et~al.(2023)Hall, Gustafson, Adcock, Misra, and
  Ross}]{hall2023zeroshot}
Melissa Hall, Laura Gustafson, Aaron Adcock, Ishan Misra, and Candace Ross.
  2023.
\newblock Vision-language models performing zero-shot tasks exhibit disparities
  between gender groups.
\newblock In \emph{Proceedings of the IEEE/CVF International Conference on
  Computer Vision (ICCV) Workshops}, pages 2778--2785.

\bibitem[{Hirota et~al.(2022)Hirota, Nakashima, and
  Garcia}]{hirota2022captionbias}
Yusuke Hirota, Yuta Nakashima, and Noa Garcia. 2022.
\newblock Quantifying societal bias amplification in image captioning.
\newblock In \emph{Proceedings of the IEEE/CVF Conference on Computer Vision
  and Pattern Recognition}, pages 13440--13449.

\bibitem[{Howard et~al.(2024)Howard, Madasu, Le, Moreno, Bhiwandiwalla, and
  Lal}]{howard2024socialcounterfactuals}
Phillip Howard, Avinash Madasu, Tiep Le, Gustavo~Lujan Moreno, Anahita
  Bhiwandiwalla, and Vasudev Lal. 2024.
\newblock Socialcounterfactuals: Probing and mitigating intersectional social
  biases in vision-language models with counterfactual examples.
\newblock In \emph{Proceedings of the IEEE/CVF Conference on Computer Vision
  and Pattern Recognition}, pages 11975--11985.

\bibitem[{Huang et~al.(2025)Huang, Qin, Zhang, Yuan, Wang, and
  Zhao}]{huang2025visbias}
Jen-tse Huang, Jiantong Qin, Jianping Zhang, Youliang Yuan, Wenxuan Wang, and
  Jieyu Zhao. 2025.
\newblock Visbias: Measuring explicit and implicit social biases in vision
  language models.
\newblock In \emph{Proceedings of the 2025 Conference on Empirical Methods in
  Natural Language Processing}, pages 17981--18004.

\bibitem[{Huben et~al.(2024)Huben, Cunningham, Smith, Ewart, and
  Sharkey}]{huben2024sae}
Robert Huben, Hoagy Cunningham, Logan~Riggs Smith, Aidan Ewart, and Lee
  Sharkey. 2024.
\newblock Sparse autoencoders find highly interpretable features in language
  models.
\newblock \emph{Proceedings of the International Conference on Learning
  Representations (ICLR)}.

\bibitem[{Jang et~al.(2025)Jang, Jung, and Wang}]{jang2025targetbias}
Taeuk Jang, Hoin Jung, and Xiaoqian Wang. 2025.
\newblock Target bias is all you need: Zero-shot debiasing of vision-language
  models with bias corpus.
\newblock In \emph{Proceedings of the IEEE/CVF International Conference on
  Computer Vision}, pages 1935--1946.

\bibitem[{Janghorbani and De~Melo(2023)}]{janghorbani2023mmbias}
Sepehr Janghorbani and Gerard De~Melo. 2023.
\newblock Multi-modal bias: Introducing a framework for stereotypical bias
  assessment beyond gender and race in vision--language models.
\newblock In \emph{Proceedings of the 17th Conference of the European Chapter
  of the Association for Computational Linguistics}, pages 1725--1735.

\bibitem[{Jung et~al.(2024)Jung, Jang, and Wang}]{jung2024sfid}
Hoin Jung, Taeuk Jang, and Xiaoqian Wang. 2024.
\newblock A unified debiasing approach for vision-language models across
  modalities and tasks.
\newblock \emph{Advances in Neural Information Processing Systems},
  37:21034--21058.

\bibitem[{Li et~al.(2023)Li, Patel, Vi{\'e}gas, Pfister, and
  Wattenberg}]{li2023inferencetime}
Kenneth Li, Oam Patel, Fernanda Vi{\'e}gas, Hanspeter Pfister, and Martin
  Wattenberg. 2023.
\newblock Inference-time intervention: Eliciting truthful answers from a
  language model.
\newblock \emph{Advances in Neural Information Processing Systems},
  36:41451--41530.

\bibitem[{Li et~al.(2025)Li, Chen, Fan, Chen, Li, Wang, and
  Zhou}]{li2025responsible}
Zhiwen Li, Die Chen, Mingyuan Fan, Cen Chen, Yaliang Li, Yanhao Wang, and
  Wenmeng Zhou. 2025.
\newblock Responsible diffusion models via constraining text embeddings within
  safe regions.
\newblock In \emph{Proceedings of the ACM on Web Conference 2025}, pages
  1588--1601.

\bibitem[{Liu et~al.(2024)Liu, Li, Li, Li, Zhang, Shen, and
  Lee}]{liu2024llavanext}
Haotian Liu, Chunyuan Li, Yuheng Li, Bo~Li, Yuanhan Zhang, Sheng Shen, and
  Yong~Jae Lee. 2024.
\newblock \href {https://llava-vl.github.io/blog/2024-01-30-llava-next/}
  {Llava-next: Improved reasoning, ocr, and world knowledge}.

\bibitem[{Liu et~al.(2023)Liu, Li, Wu, and Lee}]{liu2023visual}
Haotian Liu, Chunyuan Li, Qingyang Wu, and Yong~Jae Lee. 2023.
\newblock Visual instruction tuning.
\newblock \emph{Advances in neural information processing systems},
  36:34892--34916.

\bibitem[{Liu et~al.(2026)Liu, Tang, and Tung}]{liu2026reasoning}
Yuansen Liu, Yixuan Tang, and Anthony Kum~Hoe Tung. 2026.
\newblock Reasoning hijacking: The fragility of reasoning alignment in large
  language models.
\newblock In \emph{Proceedings of the 64th Annual Meeting of the Association
  for Computational Linguistics (Volume 1: Long Papers)}, pages 36646--36665.

\bibitem[{Manzini et~al.(2019)Manzini, Chong, Black, and
  Tsvetkov}]{manzini2019black}
Thomas Manzini, Lim~Yao Chong, Alan~W Black, and Yulia Tsvetkov. 2019.
\newblock Black is to criminal as caucasian is to police: Detecting and
  removing multiclass bias in word embeddings.
\newblock In \emph{Proceedings of the 2019 Conference of the North American
  Chapter of the Association for Computational Linguistics: Human Language
  Technologies, Volume 1 (Long and Short Papers)}, pages 615--621.

\bibitem[{Mardia and Jupp(2000)}]{mardia2000directional}
Kanti~V. Mardia and Peter~E. Jupp. 2000.
\newblock Directional statistics.
\newblock In \emph{Wiley Series in Probability and Statistics}. Wiley.

\bibitem[{Molahasani et~al.(2025)Molahasani, Motamedi, Greenspan, Kim, and
  Etemad}]{molahasani2025prism}
Mahdiyar Molahasani, Azadeh Motamedi, Michael Greenspan, Il-Min Kim, and Ali
  Etemad. 2025.
\newblock Prism: Reducing spurious implicit biases in vision-language models
  with llm-guided embedding projection.
\newblock In \emph{Proceedings of the IEEE/CVF International Conference on
  Computer Vision}, pages 688--697.

\bibitem[{Narnaware et~al.(2025)Narnaware, Vayani, Gupta, Swetha, and
  Shah}]{narnaware2026bbqvbenchmarkingvisualstereotype}
Vishal Narnaware, Ashmal Vayani, Rohit Gupta, Sirnam Swetha, and Mubarak Shah.
  2025.
\newblock \href {https://arxiv.org/abs/2502.08779} {Bbq-v: Benchmarking visual
  stereotype bias in large multimodal models}.
\newblock \emph{Preprint}, arXiv:2502.08779.

\bibitem[{Pang et~al.(2025)Pang, Qiao, Walker, Cunningham, and
  Koh}]{pang2025cabin}
Bo~Pang, Tingrui Qiao, Caroline Walker, Chris Cunningham, and Yun~Sing Koh.
  2025.
\newblock Cabin: Debiasing vision-language models using backdoor adjustments.
\newblock In \emph{Proceedings of the Thirty-Fourth International Joint
  Conference on Artificial Intelligence, IJCAI-25}, pages 484--492.

\bibitem[{Park et~al.(2023)Park, Choe, and Veitch}]{park2023linear}
Kiho Park, Yo~Joong Choe, and Victor Veitch. 2023.
\newblock The linear representation hypothesis and the geometry of large
  language models.
\newblock \emph{arXiv preprint arXiv:2311.03658}.

\bibitem[{Pennec(2006)}]{pennec2006riemannian}
Xavier Pennec. 2006.
\newblock Intrinsic statistics on {Riemannian} manifolds: Basic tools for
  geometric measurements.
\newblock \emph{Journal of Mathematical Imaging and Vision}, 25(1):127--154.

\bibitem[{Radford et~al.(2021)Radford, Kim, Hallacy, Ramesh, Goh, Agarwal,
  Sastry, Askell, Mishkin, Clark, Krueger, and Sutskever}]{radford2021learning}
Alec Radford, Jong~Wook Kim, Chris Hallacy, Aditya Ramesh, Gabriel Goh,
  Sandhini Agarwal, Girish Sastry, Amanda Askell, Pamela Mishkin, Jack Clark,
  Gretchen Krueger, and Ilya Sutskever. 2021.
\newblock \href {https://proceedings.mlr.press/v139/radford21a.html} {Learning
  transferable visual models from natural language supervision}.
\newblock In \emph{Proceedings of the 38th International Conference on Machine
  Learning}, volume 139 of \emph{Proceedings of Machine Learning Research},
  pages 8748--8763. PMLR.

\bibitem[{Ratzlaff et~al.(2024)Ratzlaff, Olson, Hinck, Tseng, Lal, and
  Howard}]{ratzlaff2024ablating}
Neale Ratzlaff, Matthew~Lyle Olson, Musashi Hinck, Shao-Yen Tseng, Vasudev Lal,
  and Phillip Howard. 2024.
\newblock Debiasing large vision-language models by ablating protected
  attribute representations.
\newblock \emph{arXiv preprint arXiv:2410.13976}.

\bibitem[{Ravfogel et~al.(2020)Ravfogel, Elazar, Gonen, Twiton, and
  Goldberg}]{ravfogel2020null}
Shauli Ravfogel, Yanai Elazar, Hila Gonen, Michael Twiton, and Yoav Goldberg.
  2020.
\newblock Null it out: Guarding protected attributes by iterative nullspace
  projection.
\newblock In \emph{Proceedings of the 58th Annual Meeting of the Association
  for Computational Linguistics}, pages 7237--7256.

\bibitem[{Ravfogel et~al.(2022)Ravfogel, Twiton, Goldberg, and
  Cotterell}]{ravfogel2022rlace}
Shauli Ravfogel, Michael Twiton, Yoav Goldberg, and Ryan Cotterell. 2022.
\newblock Linear adversarial concept erasure.
\newblock In \emph{International Conference on Machine Learning}, pages
  18400--18421.

\bibitem[{Ross et~al.(2021)Ross, Katz, and Barbu}]{ross2021groundedbias}
Candace Ross, Boris Katz, and Andrei Barbu. 2021.
\newblock Measuring social biases in grounded vision and language embeddings.
\newblock In \emph{Proceedings of the 2021 Conference of the North American
  Chapter of the Association for Computational Linguistics: Human Language
  Technologies}, pages 998--1008.

\bibitem[{Ruggeri and Nozza(2023)}]{ruggeri2023multidim}
Gabriele Ruggeri and Debora Nozza. 2023.
\newblock A multi-dimensional study on bias in vision-language models.
\newblock In \emph{Findings of the Association for Computational Linguistics:
  ACL 2023}, pages 6445--6455.

\bibitem[{Seth et~al.(2023)Seth, Hemani, and Agarwal}]{seth2023dear}
Ashish Seth, Mayur Hemani, and Chirag Agarwal. 2023.
\newblock Dear: Debiasing vision-language models with additive residuals.
\newblock In \emph{Proceedings of the IEEE/CVF Conference on Computer Vision
  and Pattern Recognition}, pages 6820--6829.

\bibitem[{Sheng et~al.(2026)Sheng, Shen, Zhao, Fang, Liu, Liang, Wang, Zhang,
  and Chua}]{sheng2026alphasteer}
Leheng Sheng, Changshuo Shen, Weixiang Zhao, Junfeng Fang, Xiaohao Liu, Zhenkai
  Liang, Xiang Wang, An~Zhang, and Tat-Seng Chua. 2026.
\newblock \href {https://openreview.net/forum?id=1vvbzAqdTe} {{AlphaSteer}:
  Learning refusal steering with principled null-space constraint}.
\newblock In \emph{The Fourteenth International Conference on Learning
  Representations}.

\bibitem[{Shoemake(1985)}]{shoemake1985slerp}
Ken Shoemake. 1985.
\newblock Animating rotation with quaternion curves.
\newblock In \emph{Proceedings of the 12th Annual Conference on Computer
  Graphics and Interactive Techniques ({SIGGRAPH})}, pages 245--254.

\bibitem[{Singh et~al.(2026)Singh, Xu, Subramanyam, and
  Kankanhalli}]{singh2026prompts}
Himanshu Singh, Ziwei Xu, AV~Subramanyam, and Mohan Kankanhalli. 2026.
\newblock Do prompts guarantee safety? mitigating toxicity from llm generations
  through subspace intervention.
\newblock \emph{arXiv preprint arXiv:2602.06623}.

\bibitem[{Srinivasan and Bisk(2022)}]{srinivasan2022dangerous}
Tejas Srinivasan and Yonatan Bisk. 2022.
\newblock Worst of both worlds: Biases compound in pre-trained
  vision-and-language models.
\newblock In \emph{Proceedings of the 4th Workshop on Gender Bias in Natural
  Language Processing (GeBNLP)}, pages 77--85.

\bibitem[{Sun et~al.(2025{\natexlab{a}})Sun, Huang, Xu, Sun, Tang, and
  Tung}]{sun2025one}
Yandong Sun, Qiang Huang, Ziwei Xu, Yiqun Sun, Yixuan Tang, and Anthony~KH
  Tung. 2025{\natexlab{a}}.
\newblock One swallow does not make a summer: Understanding semantic structures
  in embedding spaces.
\newblock \emph{arXiv preprint arXiv:2512.00852}.

\bibitem[{Sun et~al.(2025{\natexlab{b}})Sun, Huang, Tung, and
  Yu}]{sun2025prism}
Yiqun Sun, Qiang Huang, Anthony Kum~Hoe Tung, and Jun Yu. 2025{\natexlab{b}}.
\newblock Prism: A framework for producing interpretable political bias
  embeddings with political-aware cross-encoder.
\newblock In \emph{Proceedings of the 63rd Annual Meeting of the Association
  for Computational Linguistics (Volume 1: Long Papers)}, pages 27719--27733.

\bibitem[{Turner et~al.(2023)Turner, Thiergart, Udell, Leech, Mini, and
  MacDiarmid}]{turner2023activation}
Alex Turner, Lisa Thiergart, David Udell, Gavin Leech, Ulisse Mini, and Monte
  MacDiarmid. 2023.
\newblock Activation addition: Steering language models without optimization.
\newblock \emph{arXiv preprint arXiv:2308.10248}.

\bibitem[{Wang et~al.(2026)Wang, Sun, Wei, Hsieh, and
  Kawahara}]{wang2026sparse}
Hao Wang, Yiqun Sun, Pengfei Wei, Lawrence~B Hsieh, and Daisuke Kawahara. 2026.
\newblock Sparse autoencoders as plug-and-play firewalls for adversarial attack
  detection in vlms.
\newblock \emph{arXiv preprint arXiv:2605.07447}.

\bibitem[{Wang et~al.(2024)Wang, Cao, Zhang, Yuan, Shan, Chen, and
  Gao}]{wang2024vlbiasbench}
Sibo Wang, Xiangkui Cao, Jie Zhang, Zheng Yuan, Shiguang Shan, Xilin Chen, and
  Wen Gao. 2024.
\newblock Vlbiasbench: A comprehensive benchmark for evaluating bias in large
  vision-language model.
\newblock \emph{arXiv preprint arXiv:2406.14194}.

\bibitem[{Wu et~al.(2024)Wu, Liu, Wang, Yao, Deng, Lv, and
  Song}]{wu2024candidate}
Xing Wu, Kehong Liu, Jianjia Wang, Junfeng Yao, Bin Deng, Rongqi Lv, and Jun
  Song. 2024.
\newblock Candidate evaluation with multimodal data-driven for recruitment.
\newblock In \emph{International Conference on Pattern Recognition}, pages
  81--96. Springer.

\bibitem[{Zhang et~al.(2025)Zhang, Guo, and Kankanhalli}]{zhang2025jointvl}
Haoyu Zhang, Yangyang Guo, and Mohan Kankanhalli. 2025.
\newblock Joint vision-language social bias removal for clip.
\newblock In \emph{Proceedings of the Computer Vision and Pattern Recognition
  Conference}, pages 4246--4255.

\bibitem[{Zhou et~al.(2022)Zhou, Lai, and Jiang}]{zhou2022vlstereoset}
Kankan Zhou, Eason Lai, and Jing Jiang. 2022.
\newblock Vlstereoset: A study of stereotypical bias in pre-trained
  vision-language models.
\newblock In \emph{Proceedings of the 2nd Conference of the Asia-Pacific
  Chapter of the Association for Computational Linguistics and the 12th
  International Joint Conference on Natural Language Processing (Volume 1: Long
  Papers)}, pages 527--538.

\bibitem[{Zou et~al.(2023)Zou, Phan, Chen, Campbell, Guo, Ren, Pan, Yin,
  Mazeika, Dombrowski, Goel, Li, Byun, Wang, Mallen, Basart, Koyejo, Song,
  Fredrikson, Kolter, and Hendrycks}]{zou2023representation}
Andy Zou, Long Phan, Sarah Chen, James Campbell, Phillip Guo, Richard Ren,
  Alexander Pan, Xuwang Yin, Mantas Mazeika, Ann-Kathrin Dombrowski, Shashwat
  Goel, Nathaniel Li, Michael~J. Byun, Zifan Wang, Alex Mallen, Steven Basart,
  Sanmi Koyejo, Dawn Song, Matt Fredrikson, and 2 others. 2023.
\newblock Representation engineering: A top-down approach to {AI} transparency.
\newblock \emph{arXiv preprint arXiv:2310.01405}.

\end{thebibliography}

\appendix
\newpage
\section{Implementation Details}
\label{app:implementation}
\providecommand{\Log}{\operatorname{Log}}
\providecommand{\Exp}{\operatorname{Exp}}

\subsection{Notation}
\label{app:naming}
\label{app:notation}

The main text uses the compact notation $\mathbf{x}_{u,a}$, where $a$ is the
protected attribute value and $u$ collects the visual factors held fixed. In
the race-debiasing implementation, $a=r\in\mathcal{R}$ and
$u=(o,b,g)$, where $o\in\mathcal{O}$ is the occupation,
$b\in\mathcal{B}_o$ is the base identity, and $g\in\mathcal{G}$ is the
gender label. Thus
\[
\mathbf{x}_{u,a}=I_{o,b}^{(r,g)}.
\]
The corresponding activation matrix is
\[
\mathbf{H}_{u,a}
=
\mathbf{H}_{o,b}^{(r,g)}
\in\mathbb{R}^{T\times D},
\]
where $D$ is the hidden dimension, $T$ is the number of discovery visual
tokens, and the $t$-th token activation is
$\mathbf{h}_{o,b,t}^{(r,g)}\in\mathbb{R}^{D}$. The normalized pooled
representative in the main text is
\[
\mathbf{y}_{u,a}
=
\nu\left(\frac{1}{T}\sum_{t=1}^{T}\mathbf{H}_{u,a,t}\right)
=
\widehat{\bar{\mathbf{h}}}_{o,b}^{(r,g)}.
\]
The within-context Fr\'echet mean and shift are similarly related by
\[
\begin{aligned}
\mathbf{c}_u&=\boldsymbol{c}_{o,b,g},\\
\mathbf{s}_{u,a}&=\mathbf{s}_{o,b}^{(r,g)}
\quad \text{when }u=(o,b,g),\ a=r.
\end{aligned}
\]
For gender debiasing, the roles of $r$ and $g$ are exchanged: $a=g$ and
$u=(o,b,r)$. The learned basis is denoted by $\mathbf{V}_{\mathrm{bias}}$
throughout the paper. At inference time, the hooked activation tensor has shape
\[
\mathbf{Z}\in\mathbb{R}^{n_{\mathrm{batch}}\times T'\times D},
\]
where $T'$ is the sequence length at the intervention layer. We use
$n_{\mathrm{batch}}$, rather than $B$, to denote batch size. The steering
strength is denoted by $\alpha\in\mathbb{R}$.

\subsection{Implementation Details of \modelname}
\label{app:ggss_impl}

\subsubsection{Discovery data and activation extraction}
\label{app:discovery-extraction}

Discovery occupations are held out on a per-task basis so that the occupation used to score the downstream bias metric never appears in the discovery pool. For the race-based tasks (MCQ, 2AFC) we discover on $\{\texttt{cook}, \texttt{doctor}, \texttt{lawyer}, \texttt{nurse}, \texttt{teacher}\}$. For the gender-based Nurse/Doctor task we discover on $\{\texttt{cook}, \texttt{lawyer}, \texttt{teacher}\}$ and hold out both \texttt{nurse} and \texttt{doctor}. For each base identity in the active discovery pool, we use all $K=5$ race variants and both gender variants, yielding $10$ counterfactual images per identity. All images within the same counterfactual group are required to have identical pixel dimensions. This is enforced by a runtime check before activation extraction to guarantee a consistent number of visual tokens.

For each image, we run a single forward pass with the fixed discovery prompt
\begin{equation*}
\texttt{``Describe this image in detail.''}
\end{equation*}
and capture the output of the target module using \texttt{register\_forward\_hook}. If the hooked module returns a tuple, we cache the first element. Otherwise, we cache the output tensor directly. The captured tensor is reshaped into
\begin{equation*}
\mathbf{H}_{o,b}^{(r,g)} \in \mathbb{R}^{T \times D}.
\end{equation*}

When $T$ exceeds a configurable cap \texttt{max\_tokens} (default: $2048$), we draw a deterministic random subset of token indices using seed $42$. The same token indices are reused for all $(r,g)$ variants of the same base identity $b$, so token positions remain aligned across counterfactual images. \modelname\ then mean-pools the retained tokens to obtain one vector per image.

\subsubsection{Numerical spherical primitives}
\label{app:spherical-primitives}

Section~3 defines the ideal spherical operations $\nu$, $\log$, $\exp$,
and $\operatorname{Slerp}$. In Appendix A, we use
$\operatorname{Normalize}_{\varepsilon}$, $\Log$, and $\Exp$ to denote
their numerical implementations. Specifically,
\[
\operatorname{Normalize}_{\varepsilon}(\mathbf{x})
=
\frac{\mathbf{x}}{\|\mathbf{x}\|_2+\varepsilon},
\qquad
\varepsilon=10^{-10}.
\]
Inner products passed to $\arccos$ are clamped to
$[-1+10^{-7},1-10^{-7}]$. Degenerate cases with geodesic distance below
$\varepsilon_{\mathrm{geo}}=10^{-8}$ return the zero tangent vector for the
logarithmic map, and tangent vectors with norm below
$\varepsilon_{\mathrm{geo}}$ return the base point under the exponential map.
For Slerp, when the angle is below $\varepsilon_{\mathrm{geo}}$, we use the
continuous small-angle limit. All spherical operations are implemented in
vectorized form over tokens.

\subsubsection{Pooled counterfactual bias subspace discovery}
\label{app:pooled-discovery}

For each image, we first mean-pool the token activations and project onto the unit sphere:
\begin{equation}
\widehat{\bar{\mathbf{h}}}_{o,b}^{(r,g)} = \operatorname{Normalize}_{\varepsilon}\!\left(\frac{1}{T}\sum_{t=1}^{T}\mathbf{h}_{o,b,t}^{(r,g)}\right).
\end{equation}
For race debiasing, $u=(o,b,g)$ and $a=r$. Thus the main-body center $\mathbf{c}_u$ is written concretely as $\boldsymbol{c}_{o,b,g}$. For each counterfactual group $(o,b,g)$, this center is computed by iterative Fr\'echet (Karcher) mean. Starting from the normalized Euclidean mean $\boldsymbol{\eta}^{(0)}$, the update at iteration $m$ is
\begin{equation}
\begin{aligned}
\bar{\mathbf{t}}^{(m)}
&= \frac{1}{K}\sum_{r}\Log_{\boldsymbol{\eta}^{(m)}}\!\left(\widehat{\bar{\mathbf{h}}}_{o,b}^{(r,g)}\right),\\
\boldsymbol{\eta}^{(m+1)}
&= \operatorname{Normalize}_{\varepsilon}\!\left(
  \Exp_{\boldsymbol{\eta}^{(m)}}\!\bigl(\bar{\mathbf{t}}^{(m)}\bigr)
  \right).
\end{aligned}
\label{eq:appendix_karcher}
\end{equation}
The iteration terminates when $\|\bar{\mathbf{t}}^{(m)}\|_2 < 10^{-7}$ or after $100$ iterations. We denote the converged center by $\boldsymbol{c}_{o,b,g}$.

The main-body shift $\mathbf{s}_{u,a}$ is written concretely as $\mathbf{s}_{o,b}^{(r,g)}$:
\begin{equation}
\mathbf{s}_{u,a}
=
\mathbf{s}_{o,b}^{(r,g)}
=
\log_{\mathbf{c}_u}(\mathbf{y}_{u,a})
=
\Log_{\boldsymbol{c}_{o,b,g}}\!\left(\widehat{\bar{\mathbf{h}}}_{o,b}^{(r,g)}\right).
\label{eq:appendix_pooled_shift}
\end{equation}
Concatenating across occupations, identities, genders, and races yields a shift matrix
\begin{equation}
\mathbf{S} \in \mathbb{R}^{N_{\mathrm{shifts}}\times D},\qquad
N_{\mathrm{shifts}}=\sum_{o\in\mathcal{O}}|\mathcal{B}_o|\,|\mathcal{G}|\,|\mathcal{R}|.
\end{equation}
The SVD $\mathbf{S} = \mathbf{U}\boldsymbol{\Sigma}\mathbf{V}^{\top}$ yields the bias subspace
\begin{equation}
\mathbf{V}_{\mathrm{bias}} = [\boldsymbol{v}_1 \mid \cdots \mid \boldsymbol{v}_k]\in\mathbb{R}^{D\times k},
\end{equation}
with $k=K-1=4$ by default. The global reference point is the Fr\'echet mean of all pooled unit vectors, computed with the same iteration as Eq.~\eqref{eq:appendix_karcher}.

\subsubsection{Gate calibration statistics}
\label{app:gate-calibration}

For each normalized pooled discovery representative $\mathbf{c}=\mathbf{y}_{u,a}$, we compute the protected-coordinate magnitude
\[
b(\mathbf{c})
=
\|\mathbf{V}_{\mathrm{bias}}^\top\log_{\boldsymbol{\mu}}(\mathbf{c})\|_2.
\]
Equivalently, in the concrete implementation notation,
$\mathbf{c}=\widehat{\bar{\mathbf{h}}}_{o,b}^{(r,g)}$. We cache
\[
\tilde{b}=\operatorname{median}\{b(\mathbf{c})\},
\qquad
\sigma_b=\operatorname{std}\{b(\mathbf{c})\}.
\]
These scalars are loaded at inference time to compute the calibrated token gate.

\subsubsection{Inference-time geodesic-gated steering}
\label{app:inference-steering}

At inference time, a forward hook intercepts the output activation tensor $\mathbf{Z}\in\mathbb{R}^{n_{\mathrm{batch}}\times T'\times D}$, flattens it into $N=n_{\mathrm{batch}}T'$ token vectors, and processes each nonzero token independently. For token $\mathbf{h}_i$, we record its norm $r_i=\|\mathbf{h}_i\|_2$ and normalized direction $\widehat{\mathbf{h}}_i=\operatorname{Normalize}_{\varepsilon}(\mathbf{h}_i)$. The implemented token update is
\begin{equation}
\begin{aligned}
\mathbf{t}_i
&= \Log_{\boldsymbol{\mu}}(\widehat{\mathbf{h}}_i),\\
\mathbf{p}_i
&= \mathbf{V}_{\mathrm{bias}}\mathbf{V}_{\mathrm{bias}}^{\top}\mathbf{t}_i,\\
\mathbf{t}_i^{\mathrm{clean}}
&= (I-\mathbf{V}_{\mathrm{bias}}\mathbf{V}_{\mathrm{bias}}^{\top})\mathbf{t}_i.
\end{aligned}
\end{equation}
The target direction is
\begin{equation}
\widehat{\mathbf{h}}_i^{\mathrm{target}}
=
\operatorname{Normalize}_{\varepsilon}\!\left(\Exp_{\boldsymbol{\mu}}(\mathbf{t}_i^{\mathrm{clean}})\right).
\end{equation}
The per-token gate uses the discovery statistics:
\begin{equation}
\begin{aligned}
z_i
&= \frac{\|\mathbf{V}_{\mathrm{bias}}^{\top}\mathbf{t}_i\|_2 - \tilde{b}}{\sigma_b + \varepsilon},\\
g_i
&= g_{\mathrm{floor}} + (1-g_{\mathrm{floor}})\,\operatorname{sigmoid}(\kappa z_i),
\end{aligned}
\end{equation}
Since $\mathbf{V}_{\mathrm{bias}}$ has orthonormal columns, $\|\mathbf{V}_{\mathrm{bias}}^\top\mathbf{t}_i\|_2=\|\mathbf{p}_i\|_2$. For the main experiments, we use $\kappa=5$ and $g_{\mathrm{floor}}=0.3$; other $g_{\mathrm{floor}}$ values are used only in the sensitivity analysis. Slerp then rotates toward the target:
\begin{equation}
\widehat{\mathbf{h}}_i^{\mathrm{steered}}
=
\operatorname{Slerp}\!\bigl(\widehat{\mathbf{h}}_i,\widehat{\mathbf{h}}_i^{\mathrm{target}};\,\alpha g_i\bigr),
\end{equation}
and norm restoration produces the final output:
\begin{equation}
\mathbf{h}_i^{\mathrm{steered}} = r_i \,\widehat{\mathbf{h}}_i^{\mathrm{steered}}.
\end{equation}
The steered tensor is reshaped to its original layout and cast to the input dtype before being returned by the hook.

\subsubsection{Saved artifacts}
\label{app:saved-artifacts}

The discovery stage saves a \texttt{.pt} checkpoint containing $(\mathbf{V}_{\mathrm{bias}},\boldsymbol{\mu},k,\tilde{b},\sigma_b)$ together with metadata: singular values $\sigma_j$, explained-variance ratios $\rho_j=\sigma_j^2/\sum_\ell\sigma_\ell^2$, the full protected-coordinate magnitude distribution $\{b(\mathbf{c})\}$ (for diagnostic histograms, see Appendix~\ref{app:hparams}), and discovery-set sizes. This checkpoint is sufficient for inference-time loading.

\subsection{Proof of Proposition~\ref{prop:ggss-projection-norm}}
\label{app:proof-ggss-projection-norm}

\begin{proof}
By Eq.~\eqref{eq:ggss-svd-basis}, the columns of $\mathbf{V}_{\mathrm{bias}}$ are
right singular vectors of $\mathbf{S}$. Hence
$\mathbf{V}_{\mathrm{bias}}^\top\mathbf{V}_{\mathrm{bias}}=I_k$. Therefore
$\mathbf{V}_{\mathrm{bias}}\mathbf{V}_{\mathrm{bias}}^\top$ is symmetric and idempotent, so it
is the orthogonal projector onto $\operatorname{span}(\mathbf{V}_{\mathrm{bias}})$.
Consequently,
\[
I-\mathbf{V}_{\mathrm{bias}}\mathbf{V}_{\mathrm{bias}}^\top
\]
is the orthogonal projector onto
\[
\operatorname{span}(\mathbf{V}_{\mathrm{bias}})^\perp
=
\{\mathbf{z}\in\mathbb R^D:\mathbf{V}_{\mathrm{bias}}^\top\mathbf{z}=\mathbf{0}\}.
\]
Using Eq.~\eqref{eq:ggss-projection}, we have
\[
\mathbf{t}_i^{\mathrm{clean}}
=
(I-\mathbf{V}_{\mathrm{bias}}\mathbf{V}_{\mathrm{bias}}^\top)\mathbf{t}_i,
\]
which proves the projection claim.

It remains to prove norm preservation. Let
$\mathbf{p}=\widehat{\mathbf{h}}_i$, $\mathbf{q}=\widehat{\mathbf{h}}_i^{\mathrm{target}}$, and
$\theta=\arccos\langle \mathbf{p},\mathbf{q}\rangle$. If $\mathbf{p}\neq\mathbf{q}$ and $\mathbf{q}\neq -\mathbf{p}$, define
\[
\mathbf{u}
=
\frac{\mathbf{q}-\cos(\theta)\mathbf{p}}{\sin(\theta)}.
\]
Then $\mathbf{u}$ has unit norm and is orthogonal to $\mathbf{p}$. The Slerp formula can be
written as
\[
\operatorname{Slerp}(\mathbf{p},\mathbf{q};\beta_i)
=
\cos(\beta_i\theta)\mathbf{p}
+
\sin(\beta_i\theta)\mathbf{u},
\]
which has unit norm. If $\mathbf{p}=\mathbf{q}$, the convention
$\operatorname{Slerp}(\mathbf{p},\mathbf{p};\beta_i)=\mathbf{p}$ also gives a unit vector. Therefore
$\widehat{\mathbf{h}}_i^{\mathrm{steered}}$ has unit norm under the stated cases. By
Eq.~\eqref{eq:ggss-norm-restore} and $r_i=\|\mathbf{h}_i\|_2$,
\[
\|\mathbf{h}_i^{\mathrm{steered}}\|_2
=
\|r_i\widehat{\mathbf{h}}_i^{\mathrm{steered}}\|_2
=
r_i
=
\|\mathbf{h}_i\|_2.
\]
This proves the result.
\end{proof}

\subsection{Ablation Variants}
\label{app:ablations_impl}

We evaluate four SVD ablations that isolate individual design choices of \modelname.

\begin{table*}[!ht]
\centering
\setlength{\tabcolsep}{2.5pt}
\renewcommand{\arraystretch}{1.05}
\scriptsize

\providecommand{\dual}[2]{\begin{tabular}{@{}r@{}}#1\\[-1.5pt]{\tiny\color{gray!55!black}#2}\end{tabular}}
\providecommand{\dualb}[2]{\begin{tabular}{@{}r@{}}{\bfseries\boldmath #1}\\[-1.5pt]{\tiny\color{gray!55!black}#2}\end{tabular}}
\providecommand{\dualo}[2]{\begin{tabular}{@{}r@{}}{\bfseries\boldmath #1}\\[-1.5pt]{\tiny\color{gray!55!black}{\bfseries\boldmath #2}}\end{tabular}}
\providecommand{\dualob}[2]{\begin{tabular}{@{}r@{}}{\bfseries\boldmath #1}\\[-1.5pt]{\tiny\color{gray!55!black}{\bfseries\boldmath #2}}\end{tabular}}
\providecommand{\avg}[1]{#1}
\providecommand{\avgb}[1]{{\bfseries\boldmath #1}}
\providecommand{\avgo}[1]{{\bfseries\boldmath #1}}
\providecommand{\avgob}[1]{{\bfseries\boldmath #1}}

\resizebox{\textwidth}{!}{%
\begin{tabular}{@{}l rrr r rrr r rrr r rrr r@{}}
\toprule[1.2pt]
& \multicolumn{4}{c}{\textbf{Pixtral-12B}} & \multicolumn{4}{c}{\textbf{LLaVA-1.6-Vicuna-7B}} & \multicolumn{4}{c}{\textbf{LLaVA-1.6-Mistral-7B}} & \multicolumn{4}{c}{\textbf{Qwen3-VL-4B}} \\
\cmidrule(lr){2-5}\cmidrule(lr){6-9}\cmidrule(lr){10-13}\cmidrule(lr){14-17}
\textbf{Method}
  & \makecell{MCQ\\$\times 10^3 \downarrow$}
  & \makecell{2AFC\\$\downarrow$}
  & \makecell{N/D\\$\downarrow$}
  & \makecell{Avg\\$\Delta\%\downarrow$}
  & \makecell{MCQ\\$\times 10^3 \downarrow$}
  & \makecell{2AFC\\$\downarrow$}
  & \makecell{N/D\\$\downarrow$}
  & \makecell{Avg\\$\Delta\%\downarrow$}
  & \makecell{MCQ\\$\times 10^3 \downarrow$}
  & \makecell{2AFC\\$\downarrow$}
  & \makecell{N/D\\$\downarrow$}
  & \makecell{Avg\\$\Delta\%\downarrow$}
  & \makecell{MCQ\\$\times 10^3 \downarrow$}
  & \makecell{2AFC\\$\downarrow$}
  & \makecell{N/D\\$\downarrow$}
  & \makecell{Avg\\$\Delta\%\downarrow$} \\
\midrule
Baseline  & 25.75 & 0.243 & 0.425 & $0\%$ & 8.70 & 0.000 & 0.600 & $0\%$ & 26.87 & 0.000 & 0.625 & $0\%$ & 14.65 & 0.368 & 0.400 & $0\%$ \\
\midrule
\multicolumn{17}{@{}l@{}}{\emph{Hard projection, no Slerp, no gate}} \\
Pooled SVD (Eucl.)
  & \dualb{10.30}{$-60\%$}
  & \dual{0.173}{$-29\%$}
  & \dual{0.425}{$+0\%$}
  & \avg{$-30\%$}
  & \dual{9.91}{$+14\%$}
  & \dual{0.000}{--}
  & \dualb{0.025}{$-96\%$}
  & \avg{$-41\%$}
  & \dual{5.81}{$-78\%$}
  & \dual{0.000}{--}
  & \dual{0.300}{$-52\%$}
  & \avg{$-65\%$}
  & \dual{6.49}{$-56\%$}
  & \dual{0.190}{$-48\%$}
  & \dualb{0.100}{$-75\%$}
  & \avgb{$-60\%$} \\
Pooled SVD (sph.)
  & \dual{15.54}{$-40\%$}
  & \dual{0.098}{$-60\%$}
  & \dualb{0.125}{$-71\%$}
  & \avgb{$-57\%$}
  & \dualb{1.36}{$-84\%$}
  & \dual{0.000}{--}
  & \dualb{0.025}{$-96\%$}
  & \avgb{$-90\%$}
  & \dual{11.20}{$-58\%$}
  & \dual{0.000}{--}
  & \dualb{0.050}{$-92\%$}
  & \avg{$-75\%$}
  & \dual{8.37}{$-43\%$}
  & \dual{0.173}{$-53\%$}
  & \dual{0.150}{$-62\%$}
  & \avg{$-53\%$} \\
\midrule
\multicolumn{17}{@{}l@{}}{\emph{Per-token hard projection}} \\
Per-token SVD (Eucl.)
  & \dual{24.59}{$-4\%$}
  & \dual{0.172}{$-29\%$}
  & \dual{0.522}{$+23\%$}
  & \avg{$-4\%$}
  & \dual{4.26}{$-51\%$}
  & \dual{0.000}{--}
  & \dual{0.600}{$+0\%$}
  & \avg{$-25\%$}
  & \dualb{2.08}{$-92\%$}
  & \dual{0.000}{--}
  & \dual{0.625}{$+0\%$}
  & \avg{$-46\%$}
  & \dual{4.35}{$-70\%$}
  & \dual{0.245}{$-33\%$}
  & \dual{0.475}{$+19\%$}
  & \avg{$-28\%$} \\
Per-token SVD (sph.)
  & \dual{32.84}{$+28\%$}
  & \dual{0.227}{$-7\%$}
  & \dual{0.400}{$-6\%$}
  & \avg{$+5\%$}
  & \dual{2.05}{$-76\%$}
  & \dual{0.000}{--}
  & \dual{0.625}{$+4\%$}
  & \avg{$-36\%$}
  & \dual{2.74}{$-90\%$}
  & \dual{0.000}{--}
  & \dual{0.625}{$+0\%$}
  & \avg{$-45\%$}
  & \dualb{4.19}{$-71\%$}
  & \dual{0.286}{$-22\%$}
  & \dual{0.450}{$+12\%$}
  & \avg{$-27\%$} \\
\midrule
\modelname\ \textbf{(Ours)}
  & \dualo{16.85}{$-35\%$}
  & \dualob{0.096}{$-61\%$}
  & \dualob{0.125}{$-71\%$}
  & \avgo{$-55\%$}
  & \dualob{1.36}{$-84\%$}
  & \dualo{0.000}{--}
  & \dualob{0.025}{$-96\%$}
  & \avgob{$-90\%$}
  & \dualo{8.48}{$-68\%$}
  & \dualo{0.000}{--}
  & \dualob{0.050}{$-92\%$}
  & \avgob{$-80\%$}
  & \dualo{5.10}{$-65\%$}
  & \dualob{0.154}{$-58\%$}
  & \dualo{0.175}{$-56\%$}
  & \avgob{$-60\%$} \\
\bottomrule[1.2pt]
\end{tabular}%
}
\caption{\textbf{SVD ablations of \modelname\ across all models and bias tasks.} Rows remove the gate and Slerp, and per-token rows additionally replace pooled SVD with per-token SVD. Results use the same best-avg-$\alpha$ protocol as Table~\ref{tab:main_results}; per-column best values within this ablation set are bolded.}
\label{tab:svd_ablation}
\end{table*}

\paragraph{Pooled SVD (Eucl.)}
This removes both spherical geometry and the Slerp/gate. Images are Euclidean-pooled, per-group Euclidean means are subtracted, and SVD on the resulting shift matrix yields a Euclidean bias basis. Inference subtracts $\alpha\mathbf{V}_{\mathrm{bias}}\mathbf{V}_{\mathrm{bias}}^{\top}(\mathbf{h}_i-\boldsymbol{\mu})$ from every token. Because it is Euclidean, this variant does not preserve token norms.

\paragraph{Pooled SVD (sph.)}
Retains spherical geometry (Fr\'echet means, tangent shifts, exp-map projection, norm restoration) but uses hard null-space projection in tangent space:
\begin{equation*}
\mathbf{t}_i^{\mathrm{steered}} = \mathbf{t}_i - \alpha\mathbf{V}_{\mathrm{bias}}\mathbf{V}_{\mathrm{bias}}^{\top}\mathbf{t}_i.
\end{equation*}
This is identical to \modelname\ with the Slerp step replaced by hard projection and the gate disabled ($g_i\!\equiv\!1$).

\paragraph{Per-token SVD (Eucl.) and Per-token SVD (sph.)}
These variants compute per-token centers $\boldsymbol{c}_{o,b,g,t}$ and per-token shifts before SVD, preserving token position information in discovery. The inference step then applies standard null-space projection. Per-token SVD is more expensive and, on most (model, task) cells, dominated by the pooled variant.

\subsection{Baseline Methods}
\label{app:baselines_impl}

\subsubsection{Classifier-Guided Mean-Shift Correction (MeanDiff)}

We adapt classifier-based debiasing to the generative VLM setting using a single image-level representation per image. For each discovery image, the activation matrix is mean-pooled to obtain $\bar{\mathbf{h}}\in\mathbb{R}^{D}$. A race classifier is then trained on these pooled activations. We consider two classifiers: an RBF-kernel SVM (\texttt{sklearn.svm.SVC(kernel="rbf")}, $C{=}1$, $\gamma{=}\texttt{scale}$) and a multinomial logistic regression (GPU-accelerated cuML where available, otherwise scikit-learn). For the Euclidean variant, the class mean $\boldsymbol{\mu}_r$ and global mean $\boldsymbol{\mu}_{\mathrm{global}}$ yield the class-specific shift $\boldsymbol{\delta}_r = \boldsymbol{\mu}_r-\boldsymbol{\mu}_{\mathrm{global}}$. At inference time the pooled activation is classified and $\alpha\boldsymbol{\delta}_{\hat r}$ is subtracted from every token. The spherical variant uses Fr\'echet means and tangent-space subtraction with norm restoration.

\subsubsection{INLP}

We implement Iterative Null-space Projection \cite{ravfogel2020null} in both Euclidean and spherical variants. Let $\mathbf{X}\in\mathbb{R}^{N\times D}$ and $y\in\{1,\dots,K\}^N$. The iteration is $\mathbf{X}^{(0)}=\mathbf{X}$. For $m=0,1,\dots$ we fit a multinomial logistic-regression classifier $f^{(m)}$ on $\mathbf{X}^{(m)}$. If accuracy falls to within $10^{-6}$ of the uniform-chance baseline $1/K$, the iteration halts. Otherwise, the orthonormalized rows of the weight matrix are appended to the bias basis and $\mathbf{X}^{(m+1)}$ is obtained by projecting $\mathbf{X}^{(m)}$ into their null space. We use at most $10$ iterations and $C{=}1.0$. After convergence, stacking and orthonormalizing the extracted directions yields $\mathbf{V}_{\mathrm{bias}}$. The Euclidean variant applies standard null-space steering. The spherical variant normalizes, computes a Fr\'echet global mean, and projects in the tangent space before exp-mapping and norm-restoring.

\subsubsection{LEACE}

We additionally evaluate LEACE \cite{belrose2023leace}, which provides a closed-form least-squares concept eraser. Let $\boldsymbol{\Sigma}_{\mathbf{x}}$ denote the covariance of pooled discovery activations $\bar{\mathbf{h}}$ and let $\boldsymbol{\Sigma}_{\mathbf{x}\mathbf{z}}$ denote the cross-covariance with the one-hot race labels. LEACE computes a whitening $\mathbf{W}=\boldsymbol{\Sigma}_{\mathbf{x}}^{-1/2}$, takes the top-$K\!-\!1$ left singular vectors of $\mathbf{W}\boldsymbol{\Sigma}_{\mathbf{x}\mathbf{z}}$, and applies the affine projection $P(\mathbf{x}) = \mathbf{x} - \mathbf{W}^{-1}\mathbf{U}\mathbf{U}^{\top}\mathbf{W}(\mathbf{x}-\boldsymbol{\mu}_{\mathrm{global}})$. We evaluate two variants: spherical LEACE, applied in tangent space, and gated geodesic LEACE, which uses LEACE's $\mathbf{U}$ as $\mathbf{V}_{\mathrm{bias}}$ inside the \modelname\ pipeline. The latter serves as a drop-in test of whether LEACE's closed-form subspace transfers under the Slerp + gate framework.

\subsubsection{BendVLM-style Retrieval Equalization}

Following \citet{gerych2024bendvlm}, we represent each discovery image by a single \emph{global} activation $\mathbf{e}=\operatorname{vec}(\mathbf{H})\in\mathbb{R}^{TD}$ and store L2-normalized reference vectors with their race labels. We replace BendVLM's text-side augmentation with an SVD-based demographic projector computed from race-center differences. At inference time we retrieve the $10$ nearest references per race via cosine similarity, average to per-race prototypes $\boldsymbol{\pi}_r$, and solve for the minimum-norm correction that equalizes cosine distance to all prototypes. The Euclidean variant uses a linear equidistance constraint. The spherical variant lifts everything to the tangent space at the global Fr\'echet mean, solves the equidistance system there, and maps back to the sphere with norm restoration. The correction tensor is cached after the first decoding step of each image because visual tokens remain unchanged across autoregressive steps.

\subsection{Model-Specific Hook Locations and Runtime Setup}
\label{app:runtime_impl}
\label{app:hook-sites}

All methods intervene at the layer denoted by \texttt{projection-mlp2}, which resolves to the final \texttt{nn.Linear} layer of the vision-to-language projection MLP. For Pixtral-12B, this is the second linear in the multi-modal projector. For LLaVA-1.6-Vicuna-7B and LLaVA-1.6-Mistral-7B, it is \codepath{model.model.mm\_projector[2]}. For Qwen3-VL-4B-Instruct, it is \codepath{model.model.visual.merger.linear\_fc2}. Models are loaded with \texttt{device\_map="auto"} and run in \texttt{bfloat16}. All steering computations use \texttt{float32} for numerical stability, and outputs are cast back to the input dtype before being returned by the hook. All stochastic components use seed $42$. Inner products passed to $\arccos$ are clamped to $[-1{+}10^{-7},1{-}10^{-7}]$, and geodesic distances below $10^{-8}$ are treated as zero.

\section{Additional Experimental Details}
\label{app:exp_details}

\subsection{Protocols}
\label{app:protocols}

\paragraph{Sweep}
All methods are evaluated at $\alpha\in\{0.25,0.5,0.75,1.0,1.5\}$ on three bias tasks. We report the best-$\alpha$ metric per (method, model, task). Discovery occupations are held out on a per-task basis as described in \S\ref{app:ggss_impl}: race-task discovery uses $\{\texttt{cook}, \texttt{doctor}, \texttt{lawyer}, \texttt{nurse}, \texttt{teacher}\}$, and gender-task discovery uses $\{\texttt{cook}, \texttt{lawyer}, \texttt{teacher}\}$ with \texttt{nurse} and \texttt{doctor} held out.

\paragraph{Decoding}
Greedy decoding (\texttt{do\_sample=False}). Default budget \texttt{max\_new\_tokens}$=2048$. For MCQ we use $8$, since the output is constrained to a single letter.

\paragraph{CEO Salary}
We sample $N_{\mathrm{bios}}=10$ biographies, replace name placeholders with ``the candidate,'' and ask the model to output only an integer salary. Parsing first searches for a 4--9 digit integer with commas/dollar signs stripped, then falls back to any integer~$\ge 1000$. The baseline maximum salary is used as an outlier cap for steered methods. Total queries per (method,$\alpha$): $B_{\mathrm{ceo}}\!\cdot\! N_{\mathrm{bios}}\!\cdot\! 5\!\cdot\! 2$, where $B_{\mathrm{ceo}}=8$. \emph{Evidential role.} With only $10$ biographies this probe is small, and we treat it as directional, corroborating evidence, not a stand-alone result; no headline claim rests on it. On both backbones where the probe elicits varying salaries, the point estimates move toward parity under \modelname\ (Qwen3-VL gender gap $-88$k $\rightarrow$ $-55$k; Pixtral $144$k $\rightarrow$ $33$k), the same direction as the statistically significant reductions on the larger probes, but the bio-cluster bootstrap CIs are wide and overlap (Qwen3-VL: $[-213\text{k}, 0]$ unsteered vs.\ $[-165\text{k}, 0]$ steered).

\paragraph{MMStar}
MMStar is evaluated with VLMEvalKit using exact matching rather than LLM-based judging to avoid judge-dependent confounds. We report $\Delta_{\mathrm{acc}} = \mathrm{Acc}^{\mathrm{method}}(\alpha) - \mathrm{Acc}^{\mathrm{baseline}}$.

\subsection{Hyperparameters}
\label{app:hparams}

Table~\ref{tab:hparams} summarizes the three \modelname-specific hyperparameters and the ranges we explored. Best-$\alpha$ is selected once per (method, model) pair under the best-avg-$\alpha$ protocol. $\kappa$ and $g_{\mathrm{floor}}$ are \emph{fixed} across all main-table experiments.

\begin{table}[ht]
\centering
\small
\begin{tabular}{@{}lll@{}}
\toprule
Hyperparameter & Default & Range explored \\
\midrule
$\alpha$ (strength) & best-$\alpha$ & $\{0.25,0.5,0.75,1.0,1.5\}$ \\
$\kappa$ (gate sharpness) & $5$ & $\{1,2,5,10\}$ \\
$g_{\mathrm{floor}}$ (gate floor) & $0.3$ & $\{0.0,0.3,0.5\}$ \\
$k$ (subspace dim) & $K{-}1$ & fixed \\
\bottomrule
\end{tabular}
\caption{\modelname\ hyperparameters and explored ranges.}
\label{tab:hparams}
\end{table}

\begin{table}[ht]
\centering
\scriptsize
\setlength{\tabcolsep}{3.2pt}
\renewcommand{\arraystretch}{0.95}
\begin{tabular}{@{}lcccc@{}}
\toprule
\textbf{Method} & \textbf{Pixtral} & \textbf{Vicuna} & \textbf{Mistral} & \textbf{Qwen3} \\
\midrule
\textbf{INLP (Eucl.)} & 1.0 & 0.25 & 1.0 & 1.5 \\
\textbf{INLP (sph.)} & 1.0 & 1.0 & 1.0 & 1.5 \\
\textbf{MeanDiff-SVM (Eucl.)} & 0.25 & 0.5 & 0.75 & 1.0 \\
\textbf{MeanDiff-SVM (sph.)} & 0.25 & 0.25 & 0.75 & 0.5 \\
\textbf{MeanDiff-LR (Eucl.)} & 0.75 & 0.5 & 1.0 & 1.0 \\
\textbf{MeanDiff-LR (sph.)} & 0.75 & 1.0 & 0.75 & 1.0 \\
\textbf{BendVLM (Eucl.)} & 1.0 & 0.75 & 0.75 & 1.5 \\
\textbf{BendVLM (sph.)} & 0.25 & 0.5 & 0.25 & 1.5 \\
\midrule
\textbf{Pooled SVD (Eucl.)} & 0.5 & 1.0 & 0.75 & 1.5 \\
\textbf{Pooled SVD (sph.)} & 0.75 & 1.0 & 1.0 & 1.5 \\
\textbf{Per-token SVD (Eucl.)} & 0.75 & 0.25 & 0.5 & 1.0 \\
\textbf{Per-token SVD (sph.)} & 1.0 & 1.0 & 0.25 & 1.5 \\
\midrule
\textbf{\modelname} & 0.75 & 1.0 & 1.0 & 1.5 \\
\bottomrule
\end{tabular}
\caption{\textbf{Best-avg-$\alpha$ choices used in Tables~\ref{tab:main_results} and~\ref{tab:mmstar}.} Each value is the single operating point selected for that (method, model) pair.}
\label{tab:alpha_choices}
\end{table}

\subsection{Per-Model Full Results and MMStar}
\label{app:per_model}
\label{app:mmstar}
\label{app:layers}

Full per-(method, $\alpha$) tables are reported in the project repository (\texttt{results\_report.md}). Table~\ref{tab:main_results} in the main paper reports best-$\alpha$ values. Layer comparisons (other than \texttt{projection-mlp2}) were preliminary and preserved unchanged from earlier work, and we do not recommend reading headline numbers from them.

\paragraph{MMStar preservation}
Table~\ref{tab:mmstar} reports the compact MMStar summary for all paper-relevant methods at each method's best-avg-$\alpha$. The project repository contains the corresponding evaluation artifacts and per-run outputs.

\subsection{Statistical Significance Protocol}
\label{app:significance}

Decoding is greedy and deterministic with a fixed seed, so repeated runs are bit-identical; the relevant uncertainty is sampling over evaluation items. For every cell of Table~\ref{tab:main_results} we computed: (i) 95\% bootstrap confidence intervals (10{,}000 resamples), over items and over base identities, for the metric and for the paired \%-change vs.\ the unsteered baseline on the same resampled items; and (ii) paired sign-flip permutation tests between methods, which are valid regardless of the plug-in metric's small-sample bias because both methods are evaluated on the same items. MMStar is compared per-question with exact McNemar tests. Table~\ref{tab:significance} in the main text summarizes the outcomes.

The largest individual claims hold on their own: the LLaVA-Vicuna Nurse/Doctor reduction ($-96\%$) has $p<10^{-4}$ with an item-level CI of $[-118\%,-57\%]$ of the baseline gap, and the MMStar ``within $\pm 0.6$\,p.p.''\ claim is backed by per-question paired CIs (all within $\pm 1.9$\,pp) and exact McNemar tests (all $p\geq 0.09$). By contrast, INLP (sph.), the strongest bias-side competitor, significantly damages Pixtral's MMStar ($-6.1$\,pp race / $-3.3$\,pp gender, McNemar $p<10^{-4}$).

\subsection{Held-Out $\alpha$ Selection (Cross-Fitted)}
\label{app:heldout_alpha}

Probes are split by base identity into two folds; $\alpha$ is selected on one fold with the paper's best-avg-$\alpha$ rule and metrics are reported on the disjoint fold, cross-fitted, applied symmetrically to all methods. Table~\ref{tab:heldout_alpha} reports \modelname: the selected $\alpha$ matches the paper's on six of eight folds, and the test-fold reduction is negative in every fold on every model.

\begin{table}[!ht]
\centering
\setlength{\tabcolsep}{4.0pt}
\renewcommand{\arraystretch}{1.05}
\scriptsize
\begin{tabular}{@{}l c c@{}}
\toprule[1.2pt]
\textbf{Model} & \makecell{selected $\alpha$,\\fold A / B (paper $\alpha$)} & \makecell{test-fold Avg\,$\Delta\%$,\\fold A / B} \\
\midrule
Pixtral-12B      & $0.75$ / $1.0$ ($0.75$) & $-27.4\%$ / $-33.8\%$ \\
LLaVA-Vicuna-7B  & $1.0$ / $0.25$ ($1.0$)  & $-50.5\%$ / $-50.0\%$ \\
LLaVA-Mistral-7B & $1.0$ / $1.0$ ($1.0$)   & $-77.8\%$ / $-93.5\%$ \\
Qwen3-VL-4B      & $1.5$ / $1.5$ ($1.5$)   & $-34.4\%$ / $-40.8\%$ \\
\bottomrule[1.2pt]
\end{tabular}
\caption{\textbf{Held-out $\alpha$ protocol for \modelname:} $\alpha$ selected on fold A (identities 1--4) is evaluated on fold B (identities 5--8) and vice versa.}
\label{tab:heldout_alpha}
\end{table}

\subsection{Per-Dimension MMStar Results}
\label{app:mmstar_dims}

Table~\ref{tab:mmstar_dims} breaks MMStar down by capability dimension at the deployed steering strength. All six dimensions, including fine-grained perception, stay within the $\pm 6$\,pp binomial noise band ($n=250$ per dimension) of the unsteered model on all four backbones.

\begin{table*}[!ht]
\centering
\setlength{\tabcolsep}{5.0pt}
\renewcommand{\arraystretch}{1.05}
\scriptsize
\begin{tabular}{@{}l l c c c c c c@{}}
\toprule[1.2pt]
\textbf{Model} & \textbf{Steering} & CP & FP & IR & LR & math & S\&T \\
\midrule
Pixtral-12B      & race   & $68.8\!\rightarrow\!69.2$ & $44.0\!\rightarrow\!41.6$ & $68.8\!\rightarrow\!71.2$ & $57.2\!\rightarrow\!55.6$ & $45.6\!\rightarrow\!44.4$ & $37.2\!\rightarrow\!37.2$ \\
Pixtral-12B      & gender & $68.8\!\rightarrow\!68.0$ & $44.0\!\rightarrow\!44.4$ & $68.8\!\rightarrow\!68.8$ & $57.2\!\rightarrow\!55.6$ & $45.6\!\rightarrow\!45.2$ & $37.2\!\rightarrow\!37.6$ \\
LLaVA-Vicuna-7B  & race   & $59.2\!\rightarrow\!58.0$ & $32.8\!\rightarrow\!38.0$ & $46.0\!\rightarrow\!45.6$ & $31.2\!\rightarrow\!32.8$ & $27.2\!\rightarrow\!25.6$ & $27.6\!\rightarrow\!26.8$ \\
LLaVA-Vicuna-7B  & gender & $59.2\!\rightarrow\!59.2$ & $32.8\!\rightarrow\!32.4$ & $46.0\!\rightarrow\!46.8$ & $31.2\!\rightarrow\!30.4$ & $27.2\!\rightarrow\!26.8$ & $27.6\!\rightarrow\!27.2$ \\
LLaVA-Mistral-7B & race   & $62.0\!\rightarrow\!62.0$ & $34.8\!\rightarrow\!36.4$ & $48.0\!\rightarrow\!46.4$ & $35.6\!\rightarrow\!35.2$ & $27.2\!\rightarrow\!27.6$ & $23.2\!\rightarrow\!25.6$ \\
LLaVA-Mistral-7B & gender & $62.0\!\rightarrow\!62.4$ & $34.8\!\rightarrow\!34.8$ & $48.0\!\rightarrow\!47.6$ & $35.6\!\rightarrow\!34.0$ & $27.2\!\rightarrow\!27.2$ & $23.2\!\rightarrow\!24.4$ \\
Qwen3-VL-4B      & race   & $72.4\!\rightarrow\!72.4$ & $56.8\!\rightarrow\!57.2$ & $71.6\!\rightarrow\!71.2$ & $60.4\!\rightarrow\!61.2$ & $59.6\!\rightarrow\!60.0$ & $48.4\!\rightarrow\!48.0$ \\
Qwen3-VL-4B      & gender & $72.4\!\rightarrow\!72.8$ & $56.8\!\rightarrow\!56.4$ & $71.6\!\rightarrow\!72.8$ & $60.4\!\rightarrow\!61.6$ & $59.6\!\rightarrow\!60.0$ & $48.4\!\rightarrow\!49.2$ \\
\bottomrule[1.2pt]
\end{tabular}
\caption{\textbf{MMStar accuracy per capability dimension}, unsteered $\rightarrow$ \modelname\ at the deployed $\alpha$ (CP coarse perception, FP fine-grained perception, IR instance reasoning, LR logical reasoning, S\&T science \& technology).}
\label{tab:mmstar_dims}
\end{table*}

\subsection{MME Results}
\label{app:mme}

As a second capability benchmark beyond MMStar, we evaluated MME on Qwen3-VL-4B under \modelname\ at the deployed strength: the perception total moves $1693.9\!\rightarrow\!1685.6$ ($-0.4\%$) and the reasoning total $631.1\!\rightarrow\!617.1$ ($-1.8\%$), with no subcategory collapsing.

\subsection{Demographic-Recognition Probes and the $\alpha$ Dial}
\label{app:recognition}

To test whether steering debiases demographic perception or deletes it, the steered model is explicitly asked the perceived gender and race of the held-out counterfactual images. Table~\ref{tab:recognition} shows the intervention is attribute-specific: steering one attribute leaves recognition of the \emph{other} attribute intact on all four backbones. Table~\ref{tab:alpha_dial} shows the removal-vs-retention trade-off is controllable through $\alpha$: at the bias-minimizing operating point the targeted attribute's reportability on face-only close-ups is strongly reduced, which is the expected behaviour of subspace removal, while at moderate strength the model retains the attribute with much of the debiasing already realized. The knee of this trade-off is attribute-dependent, so $\alpha$ can be chosen per attribute and task; tasks that legitimately require the attribute can operate at moderate $\alpha$ or gate steering off, without retraining.

\begin{table}[!ht]
\centering
\setlength{\tabcolsep}{4.0pt}
\renewcommand{\arraystretch}{1.05}
\scriptsize
\begin{tabular}{@{}l c c@{}}
\toprule[1.2pt]
\textbf{Model} & \makecell{gender recognition,\\under race steering} & \makecell{race recognition,\\under gender steering} \\
\midrule
Pixtral-12B      & $88.7 \rightarrow 88.7$  & $70.0 \rightarrow 75.0$ \\
Qwen3-VL-4B      & $97.5 \rightarrow 96.3$  & $80.0 \rightarrow 82.5$ \\
LLaVA-Vicuna-7B  & $93.8 \rightarrow 81.2$  & $66.2 \rightarrow 66.2$ \\
LLaVA-Mistral-7B & $98.8 \rightarrow 100.0$ & $72.5 \rightarrow 72.5$ \\
\bottomrule[1.2pt]
\end{tabular}
\caption{\textbf{Recognition accuracy (\%) of the non-targeted attribute}, unsteered $\rightarrow$ steered at the paper's $\alpha$: the attribute that steering is not aimed at remains recognizable.}
\label{tab:recognition}
\end{table}

\begin{table}[!ht]
\centering
\setlength{\tabcolsep}{3.0pt}
\renewcommand{\arraystretch}{1.05}
\scriptsize
\begin{tabular}{@{}l c c c c@{}}
\toprule[1.2pt]
$\alpha$ & \makecell{race\\recog.\ (\%)} & \makecell{MCQ race bias\\(JSD$\times 10^3$)} & \makecell{gender\\recog.\ (\%)} & \makecell{Nurse/Doctor\\gap} \\
\midrule
$0$ (unsteered) & $80.0$ & $14.65$          & $97.5$ & $0.400$ \\
$0.5$           & $76.2$ & $6.54$ ($-55\%$) & $96.3$ & $0.375$ ($-6\%$) \\
$1.5$ (paper)   & $20.0$ & $5.10$ ($-65\%$) & $42.5$ & $0.175$ ($-56\%$) \\
\bottomrule[1.2pt]
\end{tabular}
\caption{\textbf{The $\alpha$ dial on Qwen3-VL-4B:} recognition of the targeted attribute vs.\ bias reduction under the corresponding steering mode. Most of the MCQ (race) reduction is realized by $\alpha=0.5$ with recognition nearly intact; the Nurse/Doctor (gender) gap responds mainly between $\alpha=0.5$ and the paper's operating point.}
\label{tab:alpha_dial}
\end{table}

\subsection{Fairness-Instruction Prompting Baseline}
\label{app:prompt_baseline}

The prompt-space baseline runs the identical probes, with no steering, after appending a fairness instruction to every prompt: ``\emph{do NOT let the person's perceived race, ethnicity, or gender influence your answer; base your answer only on task-relevant visual evidence}''. Table~\ref{tab:prompt_baseline} shows instruction prompting is inconsistent across models and tasks, and can even amplify bias, whereas \modelname\ reduces bias consistently on every model.

\begin{table*}[!ht]
\centering
\setlength{\tabcolsep}{6.0pt}
\renewcommand{\arraystretch}{1.05}
\scriptsize
\begin{tabular}{@{}l c c c c@{}}
\toprule[1.2pt]
\textbf{Model} & \makecell{MCQ salary\\(JSD$\times 10^3$)} & \makecell{2AFC income\\(std)} & \makecell{Nurse/Doctor\\gap} & \makecell{\modelname\\avg} \\
\midrule
Pixtral-12B      & $25.75\!\rightarrow\!18.88$ ($-27\%$) & $0.243\!\rightarrow\!0.281$ ($+16\%$) & $0.425\!\rightarrow\!0.425$ ($0\%$)   & $-55\%$ \\
Qwen3-VL-4B      & $14.65\!\rightarrow\!14.29$ ($-2\%$)  & $0.368\!\rightarrow\!0.299$ ($-19\%$) & $0.400\!\rightarrow\!0.450$ ($+13\%$) & $-60\%$ \\
LLaVA-Vicuna-7B  & $8.70\!\rightarrow\!16.06$ ($+85\%$)  & $0 \rightarrow 0$ (---)               & $0.600\!\rightarrow\!0.350$ ($-42\%$) & $-90\%$ \\
LLaVA-Mistral-7B & $26.87\!\rightarrow\!9.04$ ($-66\%$)  & $0 \rightarrow 0$ (---)               & $0.625\!\rightarrow\!0.575$ ($-8\%$)  & $-80\%$ \\
\bottomrule[1.2pt]
\end{tabular}
\caption{\textbf{Fairness-instruction prompting (no steering):} bias metric before $\rightarrow$ after adding the instruction (\% change). \modelname's average reduction shown for reference.}
\label{tab:prompt_baseline}
\end{table*}

\subsection{LEACE Baseline Results}
\label{app:leace_results}

LEACE \cite{belrose2023leace} is fit with the authors' official concept-erasure implementation and applied at the same projection layer as every other method (implementation details in Appendix~\ref{app:baselines_impl}), in two variants: a direct spherical port, and LEACE's closed-form subspace combined with our calibrated gate. Both were evaluated on all four backbones under the paper's best-avg-$\alpha$ protocol; Table~\ref{tab:added_baselines} in the main text reports the results. Ported this way, LEACE is a serious baseline: it wins the Pixtral column at its operating point and preserves Pixtral MMStar there ($-0.8$/$-0.9$\,pp for the two variants, McNemar $p\geq 0.15$), so its results do not stem from a broken port. It is, however, less consistent across backbones ($-26\%$ to $-31\%$ on Qwen3-VL), while \modelname\ keeps the best four-backbone mean and worst case and preserves MMStar throughout.

\subsection{Matched-Gate Geometry Ablation}
\label{app:geometry_ablation}

Table~\ref{tab:geometry_ablation} toggles Euclidean versus spherical steering at matched gate condition on all four backbones under the paper's full protocol (all bias tasks, best-avg-$\alpha$ applied identically to every variant), with paired sign-flip tests on every geometry contrast. At 4B and 7B the geometry toggle is within sampling noise in both gate conditions, and on Qwen3-VL the point estimates favor the Euclidean form; at the single Table~\ref{tab:component_ablation} cell, a Euclidean update with the same gate plus renormalization posts the lowest value we measured ($1.51$ vs.\ \modelname's $2.41$, paired $p\approx 0.75$; renormalization returns the update to the token's norm sphere, so this variant is a chordal discretization of the same on-sphere move). Geometry becomes decisive with scale: on Pixtral-12B the spherical form is significantly better under both toggles (no gate $p=0.031$; gate fixed $p=0.015$), both Euclidean forms significantly damage MMStar there ($-6.3$\,pp, McNemar $p<10^{-8}$) while both spherical forms preserve it on every backbone, and over-steering exposes failure modes the geodesic form did not show in any of our runs: the off-sphere update overshoots to roughly $9\times$ the unsteered bias at $\alpha=1.0$, and the renormalized variant degenerates outright there ($72$ of $80$ responses unparseable).

\begin{table*}[!ht]
\centering
\setlength{\tabcolsep}{6.0pt}
\renewcommand{\arraystretch}{1.05}
\scriptsize
\begin{tabular}{@{}l c c c c@{}}
\toprule[1.2pt]
\textbf{Variant} & Pixtral-12B & Qwen3-VL-4B & LLaVA-Mistral-7B & LLaVA-Vicuna-7B \\
\midrule
Euclidean hard projection (no gate) & $-29.6\%$ & $-59.7\%$ & $-65.2\%$ & $-41.0\%$ \\
Spherical hard projection (no gate) & $-56.6\%$ & $-52.8\%$ & $-75.2\%$ & $-90.1\%$ \\
\quad $p$ (geometry, no gate)       & $0.031$   & $0.30$    & $0.60$    & $0.50$    \\
\midrule
Euclidean + calibrated gate                       & $-18.3\%$ & $-62.4\%$ & $-67.6\%$ & $-86.1\%$ \\
\modelname\ (spherical + calibrated gate)         & $-55.3\%$ & $-59.9\%$ & $-80.2\%$ & $-90.1\%$ \\
\quad $p$ (geometry, gate fixed)                  & $0.015$   & $0.81$    & $0.41$    & $0.66$    \\
\bottomrule[1.2pt]
\end{tabular}
\caption{\textbf{Geometry toggled at matched gate condition}, Avg\,$\Delta\%$ under the paper's protocol (each variant at its own best-avg-$\alpha$; lower is better). ``$p$'' rows: paired sign-flip test of the geometry contrast directly above, combined across tasks. For the two LLaVA models the average covers MCQ and Nurse/Doctor only (their 2AFC baseline is zero, so $\Delta\%$ is undefined; the convention is identical for all methods, as in Table~\ref{tab:main_results}).}
\label{tab:geometry_ablation}
\end{table*}

\subsection{Gate Calibration: Pooled vs.\ Per-Token Statistics}
\label{app:gate_calibration_analysis}

The gate statistics $\tilde{b},\sigma_b$ are calibrated on pooled discovery representatives, while steering acts on individual tokens. We compared the two distributions directly under the paper's pooled subspace: the per-token bias-norm distribution is right-shifted relative to the pooled representatives (Qwen3-VL: median $0.196$ vs.\ $0.113$), so with pooled calibration the gate operates in its upper range on real tokens: mean gate $0.85$, with $72\%$ of tokens gated above $0.9$ and the least-biased $14\%$ suppressed below $0.4$. Pooled calibration therefore gives an operating point that steers most tokens and spares the low-bias tail. Re-running \modelname\ with the gate statistics recalibrated on the per-token distribution (same subspace, same $\alpha$) reaches the same bias reduction at the paper's operating point ($-65\%$ MCQ on Qwen3-VL, identical to pooled), the same or better reductions at lower $\alpha$ ($-91\%$ at $\alpha=0.5$), equally preserves MMStar ($61.8$ vs.\ baseline $61.5$, McNemar $p=0.56$), and under the full protocol ends within a point of pooled \modelname\ on Qwen3-VL ($-60.4\%$ vs.\ $-59.9\%$) and at $-84.6\%$ on LLaVA-Mistral. The choice of calibration set does not change the paper's conclusions.

\subsection{Discovery-Prompt Invariance}
\label{app:prompt_invariance}

The steered layer is the vision-to-language projection: its visual-token activations are computed before any interaction with the text prompt, so the discovery prompt architecturally cannot influence them. We verified this empirically: re-running discovery with four alternative prompts yields subspaces identical to numerical precision (principal-angle cosines $\geq 0.9997$) and identical end-to-end steering results (Table~\ref{tab:prompt_invariance}).

\begin{table}[!ht]
\centering
\setlength{\tabcolsep}{6.0pt}
\renewcommand{\arraystretch}{1.05}
\scriptsize
\begin{tabular}{@{}l c@{}}
\toprule[1.2pt]
\textbf{Discovery prompt} & steered JSD$\times 10^3$ \\
\midrule
``Describe this image in detail.'' (paper)        & $5.10$ \\
``What do you see in this image?''                & $5.10$ \\
``Caption this image.''                           & $5.10$ \\
``Briefly describe the person in the photo.''     & $5.10$ \\
``Write a short story inspired by this image.''   & $5.10$ \\
\bottomrule[1.2pt]
\end{tabular}
\caption{\textbf{Discovery-prompt invariance:} Qwen3-VL-4B MCQ race bias (JSD$\times 10^3$, unsteered $14.65$) when $\mathbf{V}_{\mathrm{bias}}$ is discovered with different prompts; steering applied at the paper's operating point.}
\label{tab:prompt_invariance}
\end{table}

\subsection{Artifact Documentation, Licensing, and Intended Use}
\label{sec:artifact_statement}

We use existing artifacts only for research evaluation and do not redistribute third-party datasets or model weights. REFLECT/FOCUS is used for controlled race- and gender-counterfactual bias evaluation; its repository is released under GPL-3.0 and documents 480 face-only counterfactual images across six occupations, eight source identities per occupation, and five race by two gender variants. SocialCounterfactuals is used consistently with its intended purpose of probing social bias in VLMs and is released under the MIT license. MMStar is used only as a general VLM capability benchmark through VLMEvalKit; VLMEvalKit is Apache-2.0 licensed. For MMStar, we rely on the original benchmark release and do not redistribute its data, since we could not verify an explicit dataset license from the public repository or dataset card.

Our released artifacts consist of GGSS code, configuration files, steering checkpoints, and aggregate evaluation outputs. These artifacts are intended for research on inference-time debiasing and auditing of frozen generative VLMs, not for making deployment claims without additional human review and domain-specific safety testing. The artifact documentation specifies the evaluated models, hook locations, bias tasks, capability benchmark, demographic attributes, occupations, held-out discovery/evaluation splits, and hyperparameter settings.

\end{document}